%% file: TCST_Kite_Retraction14_arXiv.tex
\documentclass[10pt,letterpaper]{article}

\usepackage{amssymb}                    
\usepackage{amsmath}                    
\usepackage{mathrsfs}                   
\usepackage{IEEEtrantools}
\usepackage{times}                      
\usepackage[T1]{fontenc}

\usepackage[pdftex]{color,graphicx}
\usepackage{ctable}                     
\usepackage[binary-units=true]{siunitx} 
\usepackage[english=american]{csquotes}	

\usepackage[font+=small]{caption}
\usepackage[font+=footnotesize]{subcaption}

\usepackage{cite}
\usepackage{url}

\newcommand{\secref}[1]{Section~\ref{#1}}
\newcommand{\figref}[1]{Fig.~\ref{#1}}
\newcommand{\figsref}[1]{Figs.~\ref{#1}}
\newcommand{\tabref}[1]{Tab.~\ref{#1}}	
\newcommand{\asref}[1]{Assumption~\ref{#1}}

\newtheorem{proof}{Proof}
\newtheorem{proposition}{Proposition}
\newtheorem{assumption}{Assumption}

\newcommand{\airush}{\textsc{Airush One}$^\text{\textregistered}$~}

\renewcommand{\vec}[1]{\mathbf{#1}}
\newcommand{\kW}[1]{\SI{#1}{\kilo\watt}}
\newcommand{\degr}[1]{\SI{#1}{\degree}}
\newcommand{\radi}[1]{\SI{#1}{\radian}}
\newcommand{\seconds}[1]{\SI{#1}{\second}}
\newcommand{\minu}[1]{\SI{#1}{\minute}}
\newcommand{\Meter}[1]{\SI{#1}{\meter}}
\newcommand{\sqm}[1]{\SI{#1}{\meter\squared}}
\newcommand{\mps}[1]{\SI[per-mode=symbol]{#1}{\meter\per\second}}
\renewcommand{\theta}{\vartheta}
\renewcommand{\phi}{\varphi}

\begin{document}
\date{August 28, 2014}

\title{Automatic Retraction and Full Cycle Operation for a Class of Airborne Wind Energy Generators
\thanks{This manuscript is a preprint of a paper submitted for possible publication on the IEEE Transactions on Control Systems Technology and is subject to IEEE Copyright. If accepted, the copy of record will be available at \textbf{IEEEXplore} library: http://ieeexplore.ieee.org/.}
\thanks{This research has received partial funding from the Swiss Competence Center Energy and Mobility (CCEM).}}
\author{A.~U. Zgraggen\thanks{Corresponding author: zgraggen@control.ee.ethz.ch.}, L. Fagiano
 \thanks{L. Fagiano is with ABB Switzerland Ltd., Corporate Research, Baden-D\"{a}ttwil, Switzerland. \protect E-mail: lorenzo.fagiano@ch.abb.com}
 and M. Morari
\thanks{A. Zgraggen and M. Morari are with the Automatic Control Laboratory, ETH Z\"{u}rich, Switzerland. \protect E-mail: zgraggen\textbar morari@control.ee.ethz.ch}}
\maketitle

\input{abstract}

\input{introduction}

\input{systemdescription}

\input{autoretraction}

\input{results}

\input{conclusion}


\appendix
\input{appendix}

\bibliographystyle{IEEEtran}

\end{document}

%% file: abstract.tex
\begin{abstract}
Airborne wind energy systems aim to harvest the power of winds blowing at altitudes higher than what conventional wind turbines reach.
They employ a tethered flying structure, usually a wing, and exploit the aerodynamic lift to produce electrical power.
In the case of ground-based systems, where the traction force on the tether is used to drive a generator on the ground, a two phase power cycle 
is carried out: one phase to produce power, where the tether is reeled out under high traction force, and a second phase where the tether is recoiled under minimal load. The problem of controlling a tethered wing in this second phase, the retraction phase, is addressed here, by proposing two possible control strategies.
Theoretical analyses, numerical simulations, and experimental results are presented to show the performance of the two approaches.
Finally, the experimental results of complete autonomous power generation cycles are reported and compared with first-principle models.
\end{abstract}

%% file: introduction.tex
\section{Introduction}
Airborne wind energy (AWE) systems are an emerging technology to harvest renewable energy from wind. Their aim is to harness the energy contained in the strong and steady winds beyond the altitude reached by traditional wind turbines, see \cite{AWE13,FaMi12} for an overview. These systems consist of a ground unit (GU), a wing, and one or more tethers connecting them.

During power production, the wing is flown in a \enquote{crosswind pattern}, i.e. roughly perpendicular to the wind flow, exceeding the speed of the wind and thus exerting high aerodynamic forces.
The generators can either be placed on-board of the wing or on the ground inside the GU. On-board generation systems use propellers driven by the high apparent wind speed and then transfer the produced power to the ground via an electrified tether, see e.g. \cite{Makani}. 
On the other hand, ground-based generation systems use the traction force on the cable to spin a drum installed on the GU and connected to a generator, see e.g. \cite{Ampyx}. In this paper we consider the latter approach. 

The wing's path can be influenced by means of different technical solutions, which typically give rise to a steering input corresponding to a change of the roll angle of the wing. 
Assuming a straight tether, the path of the wing is restricted to a spherical surface with a radius equal to the tether length, confined by the ground and a vertical plane perpendicular to the wind direction. This spherical surface is commonly called \enquote{wind window}.
Depending on the path flown by the wing, a higher or lower traction force is experienced. During crosswind paths a high wing speed can be achieved and thus a high traction force is exerted. On the other hand, if the wing is flown on the side of the wind window, i.e. with the tether roughly perpendicular to the wind direction, a low wing speed results and a small traction force is exerted.

These two different flying conditions are exploited in ground-based generation AWE systems by flying a two phase power cycle \cite{Loyd80,CaFM09c}. In the first phase, called traction phase, power is produced by flying a crosswind pattern and using the high traction force to unreel the tether from the drum. Once the maximum tether length is reached, the second phase, called retraction phase, is carried out by moving the wing on the side of the wind window and then recoiling the lines under low traction forces.
In this way only a fraction of the energy previously produced is consumed.
This approach is considered by various companies and research groups \cite{Ampyx,SwissKitePower,TwingTec,Enerkite,Kitenergy,SkySailsPower,WindLift,eKite,Argatov2010_Struct,JehleSchmehl14,CaFM07,HouDi10,TBSO11}.

The automatic control of tethered wings plays a major role for the operation of this kind of system and has been studied in various publications, see \cite{CaFM09c,IlHD07,BaOc12,WiLO08,HouDi07,CoBo13,PhDMoritzDiehl,JehleSchmehl14,FaZg13,ErSt12}. 
Several of these approaches consider only the problem of flying crosswind trajectories when energy is produced. However, for ground-based generation systems also the retraction of the tether has to be done autonomously.
In \cite{CaFM09c} and \cite{IlHD07} two controllers for the retraction phase, using nonlinear Model Predictive Control strategies, have been proposed.
However, these strategies might be difficult to implement and tune due to their complexity. Additionally, they assume quite a good knowledge of the wind speed at the wing's location, which is hard to obtain in practice, and they have been tested in simulations only, assuming that the model used for the control calculation corresponds exactly to the real system.

In this paper, we tackle the problem of autonomous retraction phase for ground-based AWE systems by presenting two possible control approaches, which we tested in real-world experiments with a small-scale prototype. The first one is an extension of the approach presented in \cite{FaZg13} and it is based on the notion of the velocity angle of the wing, which represents its flying direction.
As we will show in this paper, this notion can be adapted such that it can also be used for feedback control during the retraction phase, when the speed of the wing relative to the GU is low and the original definition of the velocity angle is not valid anymore. The resulting controller is dependent on an estimate of the wind direction at the wing's location. We will show that with this approach the wing can be stabilized at the border of the wind window during the retraction phase. 

Since an estimate of the wind direction at the wing's location is not straightforward to obtain, we propose an alternative approach by controlling directly the elevation of the wing. In order to do so, we derive a new model relating the steering input to the vertical acceleration of the wing, and we use such a model for control design.
Also this control system is able to stabilize the wing's trajectory at a constant elevation angle and it only relies on directly measurable variables, hence resulting in a more reliable and robust solution with respect to the previous one. 
The considerations above are supported by simulation results used to compare the two approaches. Real-world experiments are then used to validate both control strategies.
There exists evidence in the literature \cite{AWE13} that other groups and companies have been flying autonomous power cycles, however there are no publications explaining the employed control strategy.
By achieving an autonomous retraction phase, we have been able to test fully autonomous power cycles in experiments, whose results we compare here to the well-known equations \cite{Loyd80} that lie at the very foundations of the concept of airborne wind energy, showing indeed a good matching between mathematical models and real-world data.

The paper is structured as follows. In \secref{sec:SystemDescription}, we describe the considered system and the models we use.
In \secref{sec:Automatic_Retraction} we introduce the two different control approaches for the retraction phase and discuss the tether reeling scheme.
In \secref{sec:results}, simulation and experimental results are presented and discussed for both control approaches.
Conclusions and future developments are given in \secref{sec:conclusion}.

%% file: systemdescription.tex
\section{System Description}
\label{sec:SystemDescription}

The system under consideration is related to the Swiss Kite Power prototype \cite{SwissKitePower}, see \figref{fig:GS1_front}. It is an AWE system featuring ground-based steering actuators with the generators placed inside the GU. It has three drums with a motor connected to each one, and it can be used with one, two, and three-line wings or power kites. 
In three-line systems, the line wound around the middle drum, called main line, is connected to the leading edge of the wing and sustains the main portion of the traction force. The lines on the other two drums are called steering lines and are connected to the left and right wing tips. These two lines are used to influence the wing's trajectory. By changing the difference, $\delta$, between the length of the two steering lines the required steering deviation can be issued. A shorter left line induces a counter-clockwise turn of the wing as seen from the GU, and vice-versa. The system has a total rated power of \kW{20}; the generator of the middle drum has a power rating of \kW{10} and each of the motors connected to the drums of the steering lines has a power rating of \kW{5}. The system is operated with tether lengths up to \Meter{200}. 
We first recall a dynamical model of the described system, followed by the definition of the velocity angle $\gamma$, which acts as one of the main feedback variables during the traction phase (for the details on the controller employed in this phase, we refer the reader to \cite{FaZg13}).

\begin{figure}[htb]
 \begin{center}
  \includegraphics[trim= 0cm 0cm 0cm 0cm,width=0.5\textwidth]{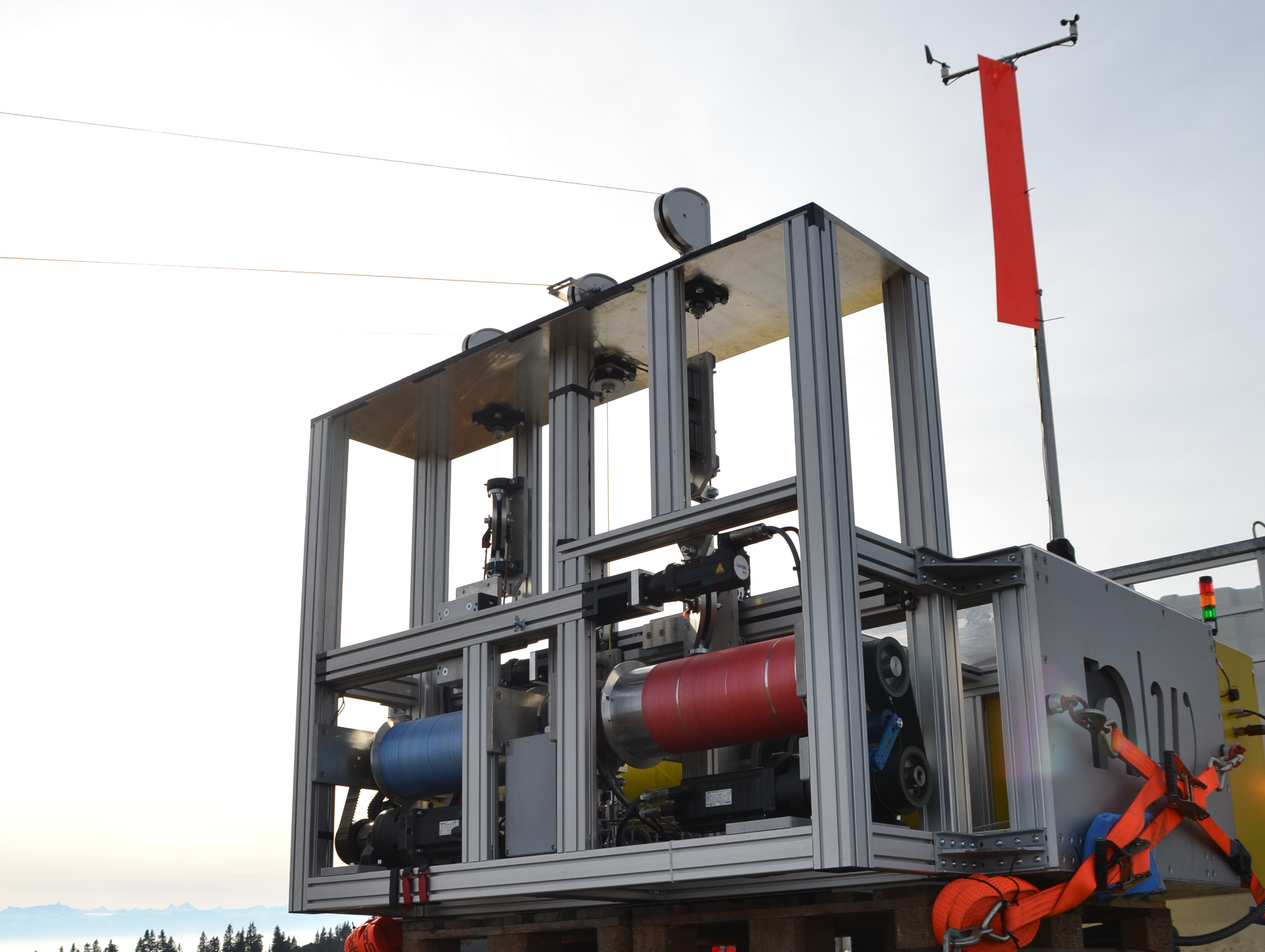}
  \caption[]{Front view of the Swiss Kite Power prototype built at Fachhochschule Nordwestschweiz. The two steering lines, left (red) and right (blue), are wound around drums connected via a belt drive to motors mounted below the drums. The center line (yellow) wound around the main drum is behind the two other drums and is partly visible below the left steering line's drum. All three lines are guided separately via pulleys to the lead-out sheaves, visible at the top. On the lead-out sheave of the main line a line angle sensor is mounted. A wind sensor, mounted roughly \Meter{5} above the ground, is visible in the background.}
  \label{fig:GS1_front}
 \end{center}
\end{figure}

\subsection{Model Equations}
The dynamical model we consider has been widely used in previous works, see e.g. \cite{CaFM09c} and references therein. We will recall the model equations shortly, following the same notation as in \cite{FaZg13} and additionally include a further degree of freedom to account for the reeling capabilities of the considered prototype. We will denote vector valued variables in bold, e.g. ${}_G\mathbf{p}(t)$, where the subscript letter in front of vectors denote the reference system considered to express the vector components and $t$ denotes the time dependence.

An inertial frame centered at the GU is denoted as $G\doteq(\mathbf{e}_x,\mathbf{e}_y,\mathbf{e}_z)$, where unit vectors are denoted by $\mathbf{e}$ with the corresponding direction indicated by the trailing subscript letter.
The $\mathbf{e}_x$ axis is assumed to be parallel to the ground, contained in the longitudinal symmetry plane of the GU, the $\mathbf{e}_z$ axis is perpendicular to the ground pointing upwards, and the $\mathbf{e}_y$ axis is such that it forms a right hand system. The wing's position vector ${}_G\mathbf{p}(t)$ can be expressed in the inertial frame using spherical coordinates  $(\phi(t),\theta(t),r(t)))$ as (see \figref{fig:GeneralSystem}):
\begin{IEEEeqnarray}{rCl}\label{eqn:WingPos}
 {}_G\mathbf{p}(t) &=& \begin{pmatrix}
                        r(t)\cos{(\phi(t))}\cos{(\theta(t))}\\
                        r(t)\sin{(\phi(t))}\cos{(\theta(t))}\\
                        r(t)\sin{(\theta(t))}
                       \end{pmatrix}\, .
\end{IEEEeqnarray}
Note that all three variables $(\phi(t),\theta(t),r(t)))$ can be measured directly with good accuracy by devices installed on the ground such as line angle sensors and motor encoders.

\begin{figure}[htb]
 \begin{center}
  \includegraphics[trim= 0cm 0cm 0cm 0cm,width=0.5\textwidth]{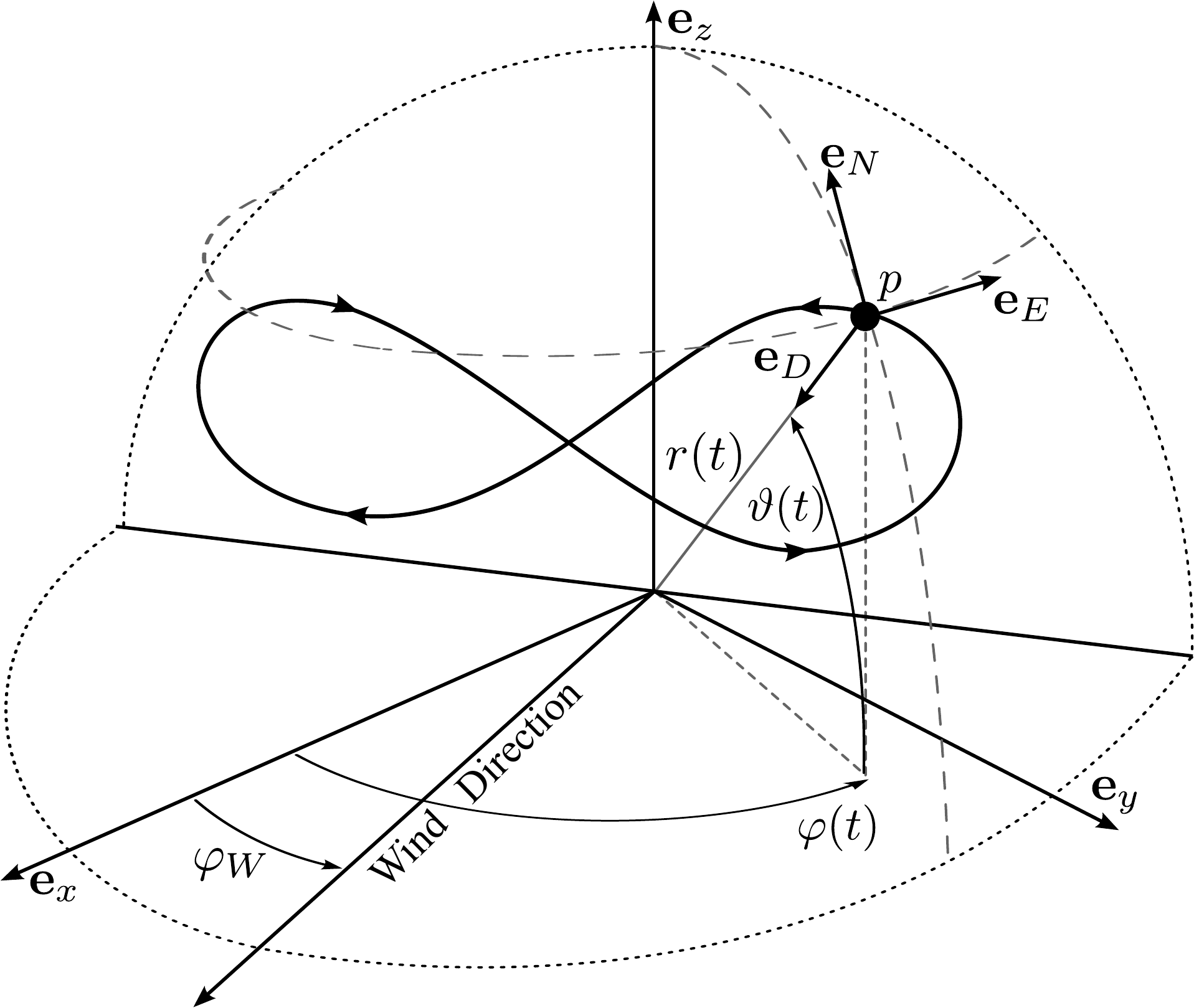}
  \caption[]{The wing's position $p$ (black dot) is shown on a figure-eight crosswind path together with the local coordinate frame $L$ and the inertial coordinate frame $G$. The wind direction is denoted by angle $\phi_W$ defining the wind window (dotted). Note the arrows on the figure-eight path showing an ``up-loop'' pattern, i.e the wing is flying upwards on the side of the path and downwards in the middle.}
  \label{fig:GeneralSystem}
 \end{center}
\end{figure}

The motion of the tethered wing is restricted to the wind window, a surface with (time-varying) radius $r(t)$ confined by the ground plane $(\mathbf{e}_x,\mathbf{e}_y)$ and by a vertical plane perpendicular to the wind direction $\phi_W$ and containing the origin of $G$.
If $r(t)$ is kept constant, the wind window corresponds to a quarter sphere. Otherwise, depending on the reeling speed $\dot{r}(t)$ of the tether, the wind window contains a larger or smaller surface area than a quarter sphere. For example with $\dot{r}(t)<0$, i.e. reeling-in the tether, the wing is able to surpass the GU against the wind direction, thanks to the additional apparent wind speed induced by the reeling.

Additionally, we define a non-inertial coordinate system $L\doteq(\mathbf{e}_N,\mathbf{e}_E,\mathbf{e}_D)$, centered at the wing's position (depicted in \figref{fig:GeneralSystem}). The $\mathbf{e}_N$ axis, or local north, is tangent to the sphere of radius $r(t)$, on which the wing's trajectory evolves, and points towards its zenith. The $\mathbf{e}_D$ axis, called local down, points to the center of the sphere (i.e the GU), hence it is perpendicular to the tangent plane of the sphere at the wing's position. The $\mathbf{e}_E$ axis, named local east, forms a right hand system and spans the tangent plane together with $\mathbf{e}_N$. The system $L$ is a function of the wing's position only. The transformation matrix to express the vectors in the local frame $L$ from the inertial frame $G$ is denoted by $A_{LG}$ (e.g. ${}_L\mathbf{p}(t) = A_{LG}\:{}_G\mathbf{p}(t)$):
  
\small
\begin{IEEEeqnarray}{rCl}\label{eqn:A_LG}
  A_{LG} &=& \begin{pmatrix}
              -\cos{(\phi)}\sin{(\theta)} & -\sin{(\phi)}\sin{(\theta)} & \cos{(\theta)}\\
              -\sin{(\phi)}               &  \cos{(\phi)}               & 0\\
              -\cos{(\phi)}\cos{(\theta)} & -\sin{(\phi)}\cos{(\theta)} & -\sin{(\theta)}
             \end{pmatrix}\, .
\end{IEEEeqnarray}
\normalsize

From the differentiation of \eqref{eqn:WingPos} and using the rotation matrix \eqref{eqn:A_LG} we obtain the velocity vector of the wing in local coordinates $L$ with respect to the GU:
\begin{IEEEeqnarray}{rCl}\label{eqn:L_Vp}
 {}_L\mathbf{v}_P(t) &=& \begin{pmatrix}
                          r(t)\dot{\theta}(t)\\
                          r\cos{(\theta(t))}\, \dot{\phi}(t)\\
                          -\dot{r}(t)
                         \end{pmatrix}\, .
\end{IEEEeqnarray}

A dynamic model of the described system can be derived from first principles, where the wing is assumed to be a point with given mass. The tether is assumed to be straight with a non-zero diameter. The aerodynamic drag of the tether and the tether mass are added to the wing's drag and mass, respectively. 
The effects of gravity and inertial forces are also considered. The wing is assumed to be steered by a change of the roll angle $\psi(t)$, which is manipulated by the control system via the line length difference $\delta(t)$.
By applying Newton's law of motion to the wing in the reference system $L$ we obtain:
\begin{IEEEeqnarray}{rCl}\label{eqn:EqnMot}
 \renewcommand\arraystretch{1.7}
 \begin{vmatrix}
   \ddot{\theta} &=& \frac{\mathbf{F}\cdot \mathbf{e}_{N}}{rm} - \sin{(\theta)} \cos{(\theta)} \dot{\phi}^2-\frac{2}{r}\dot{\theta}\dot{r}\\
   \ddot{\phi}   &=& \frac{\mathbf{F}\cdot \mathbf{e}_{E}}{rm\cos{(\theta)}} + 2\tan{(\theta)}\dot{\theta}\dot{\phi}-\frac{2}{r}\dot{\phi}\dot{r}\\
   \ddot{r}      &=& -\frac{\mathbf{F}\cdot \mathbf{e}_{D}}{m} + r\dot{\theta}^2 + r\cos^2{(\theta)} \dot{\phi}^2
 \end{vmatrix}\, ,
\end{IEEEeqnarray}
where $m$ is the mass of the wing. The force $\mathbf{F}(t)$ consist of contributions from gravity $\vec{F}_g(t)$, aerodynamic force $\vec{F}_a(t)$, and the force exerted by the lines $\vec{F}_c(t)$. Note that for simplicity of notation we dropped the time dependence of the involved variables in \eqref{eqn:A_LG} and \eqref{eqn:EqnMot}.
The force $\vec{F}_c(t)$, called traction force, opposes all other forces along the tether direction and can be influenced by the motors in the GU to control the tether reeling.
For details on the derivation of the involved forces, see \cite{CaFM09c}.
Equations \eqref{eqn:EqnMot} gives an analytic expression for the point-mass model of the wing with six states, $(\phi(t),\theta(t),r(t),\dot{\phi}(t),\dot{\theta}(t),\dot{r}(t))$, two manipulated inputs $(\delta(t),|\vec{F}_c(t)|)$, and three exogenous inputs stemming from the vector components of the incoming wind $\vec{W}(t)$. Such a model has been widely used in literature for the control design of airborne wind energy systems, see e.g. \cite{CaFM09c,CaFM07,IlHD07,BaOc12}.

In a recent contribution \cite{FaZg13} concerned with the autonomous flight along figure-eight paths during the traction phase, the notion of the velocity angle $\gamma$ has been introduced:
\begin{IEEEeqnarray}{rCl}
 \gamma(t) &\doteq& \arctan{\left(\frac{\mathbf{v}_P(t)\cdot\mathbf{e}_{E}(t)}{\mathbf{v}_P(t)\cdot\mathbf{e}_{N}(t)}\right)}\label{eqn:Gamma_definition} \\
           &=& \arctan{\left(\frac{\cos{(\theta(t))}\, \dot{\phi}(t)}{\dot{\theta}(t)}\right)}\label{eqn:Gamma}\, .
\end{IEEEeqnarray}
Thus, $\gamma(t)$ is the angle between the local north $\mathbf{e}_{N}(t)$ and the projection of the wing's velocity vector $\mathbf{v}_P(t)$ onto the tangent plane spanned by the local north and east vectors.
In \eqref{eqn:Gamma} the four-quadrant version of the arc tangent function shall be used, such that $\gamma(t)\in[-\pi,\pi]$.

The velocity angle describes the flight conditions of the wing with just one scalar: as an example, if $\gamma=0$ the wing is moving upwards towards the zenith of the wind window, and if $\gamma=\pi/2$ the wing is moving parallel to the ground towards the local east.
Additionally, a control-oriented model for tethered wings, originally proposed in \cite{ErSt12} and refined in \cite{FaZg13}, has been used for the control design of the traction phase:
\begin{IEEEeqnarray}{rCl}\label{eqn:GammaDotLaw}
 \dot{\gamma}(t) &\simeq& K(t)\delta(t)+T(t)\, ,
\end{IEEEeqnarray}
where
\begin{IEEEeqnarray}{rCl}
 K(t) &=& \frac{\rho C_L(t)A|\vec{W}_a(t)|}{2md_s}\left(1+\frac{1}{E_{eq}^2(t)}\right)^2\IEEEyesnumber\label{eqn:GammaDotLawCoeffs}\IEEEyessubnumber \label{eqn:GammaDotLawCoeffs1}\\
 T(t) &=& \frac{g\cos{(\theta(t))}\, \sin{(\gamma(t))}}{|\vec{W}_a(t)|}+\sin{(\theta(t))}\, \dot{\phi}(t) \IEEEyessubnumber \label{eqn:GammaDotLawCoeffs2}
\end{IEEEeqnarray}
In \eqref{eqn:GammaDotLaw} and \eqref{eqn:GammaDotLawCoeffs} the steering input, i.e. the line length difference of the steering lines, is denoted by $\delta(t)$, $\rho$ is the air density, $C_L(t)$ is the aerodynamic lift coefficient, $A$ is the reference area of the wing, $d_s$ is the span of the wing, $E_{eq}(t)$ is the equivalent efficiency of the wing, defined as $E_{eq}(t)\doteq C_L(t)/C_{D,eq}(t)$, where $C_{D,eq}(t)$ represents the drag coefficient of the wing and lines together, and $g$ is the gravitational acceleration.
The apparent wind $\vec{W}_a(t)$ is defined as
\begin{IEEEeqnarray}{rCl}\label{eqn:WapparentDefinition}
 \vec{W}_a(t) &=& \vec{W}(t) - \vec{v}_p(t)\, , 
\end{IEEEeqnarray}
where the incoming wind $\vec{W}(t)$ in the $L$ frame can be written as
\begin{IEEEeqnarray}{rCl}\label{eqn:L_W}
  {}_L\vec{W}(t) &=& \begin{pmatrix}
                       -W_0\cos{(\phi(t)-\phi_W)}\sin{(\theta(t))}\\
                       -W_0\sin{(\phi(t)-\phi_W)}\\
                       -W_0\cos{(\phi(t)-\phi_W)}\cos{(\theta(t))}
                      \end{pmatrix}
\end{IEEEeqnarray}
with $W_0$ being the nominal wind speed (which can eventually also be position dependent, if a wind shear model is considered).

The model \eqref{eqn:GammaDotLaw} has been validated through experimental data at constant line length with good correspondence in a wide range of operating conditions, see \cite{FaZg13}. 
It was derived assuming crosswind flight conditions as performed during the traction phase.

During retraction, the tethers have to be recoiled onto the drums under minimal force, such that only a small fraction of the previously generated energy is used.
To achieve this goal the wing is flown at the border of the wind window, in a static angular position w.r.t. the GU, i.e. with constant or slowly varying $\phi$ and $\theta$ angles. 
This represents quite a different flight condition with respect to the one assumed in \eqref{eqn:GammaDotLaw}.
However as discussed in the next section, it can been shown that the model \eqref{eqn:GammaDotLaw} can also, with some modifications, be used to describe the wing's steering dynamics during the retraction phase, employing a slightly changed definition of the velocity angle \eqref{eqn:Gamma} called ``regularized velocity angle''.

%% file: autoretraction.tex
\section{Automatic Retraction of Ground-Based Airborne Wind Energy Systems}
\label{sec:Automatic_Retraction}
The control problem of automatically retracting the wing during the retraction phase involves two tasks; reeling the tether on the drum and stabilizing the position of the wing at the border of the wind window.
One of our contribution is to show that the dynamics of the wing position during retraction are almost linear thus standard linear control techniques can be applied. Additionally, the reeling control can be considered as a decoupled problem which influences the position control system as a disturbance.

We will present two different control strategies for the problem of stabilizing the wing's position during the retraction phase. The first approach uses a regularized version of the notion of the velocity angle employed by the traction phase controller and is presented in \secref{ssec:Velocity_Angle_Based_Retraction}. It needs only minor changes from the traction phase controller but relies on estimates of the wind direction and speed at the wing's location.
The second approach, explained in \secref{ssec:Elevation_Dynamics_Based_Retraction}, is based on a simplified model, introduced in this work, of the elevation dynamics of the tethered wing during the retraction. 
This second approach has the advantage of employing only directly measurable quantities, hence it does not need an estimate of the wind direction nor of the velocity angle.
In \secref{ssec:RetrCtrl_Discussion} we highlight the connections between the two approaches and in \secref{ssec:Retr_Reeling} we discuss the reeling strategy.

\subsection{Rectraction Control Based on Regularized Velocity Angle}
\label{ssec:Velocity_Angle_Based_Retraction}


The main difference between the traction and retraction phases is the magnitude of the wing's speed in the tangent plane to the wind window at the wing's location, denoted by $\mathbf{v}_P^p(t)$. 
During the retraction phase, $\mathbf{v}_P(t)$ is low and mainly consists of the reel-in speed $\dot{r}(t)$. Thus, $\mathbf{v}_P^p(t)$ is close to zero and the apparent wind speed is determined only by the wind speed $\vec{W}(t)$ and the reel-in speed $\dot{r}(t)$. 
In these conditions, the velocity angle $\gamma$ as computed in \eqref{eqn:Gamma} becomes undefined, so that this variable cannot be used for feedback control anymore.

Recall that we assume for simplicity that the wind flow is parallel to the ground, i.e. the $(\vec{e}_x,\,\vec{e}_y)$ plane, and its direction forms an angle $\phi_W$ w.r.t. $\vec{e}_x$ (see \figref{fig:GeneralSystem}). It is also assumed that the wing is designed so that it orientates itself into the apparent wind, which means that the wing's longitudinal symmetry axis is aligned with the vector $\vec{W}_a(t)$ \eqref{eqn:WapparentDefinition}, i.e. the wind direction during retraction, projected onto the tangent plane to the wind window at the wing's location. This effect can be achieved by a wing equipped with a rudder or a curved shape, like C-shaped surf kites.
Thus, during retraction the component of $\vec{W}_a(t)$ in the tangent plane to the wind window can be assumed to be equal to the wind velocity projected on the same plane.

Under this assumption, we can compute the orientation $\beta(t)$ of the wing using \eqref{eqn:L_W}, as
\begin{IEEEeqnarray}{rCl}
 \beta(t) &\doteq& \arctan{\left(\frac{-{}_L\vec{W}(t)\cdot\mathbf{e}_{E}(t)}{-{}_L\vec{W}(t)\cdot\mathbf{e}_{N}(t)}\right)}\label{eqn:beta_definition} \\
          &=&      \arctan{\left(\frac{\sin{(\phi-\phi_W)}}{\sin{\theta}\cos{(\phi-\phi_W)}}\right)}\label{eqn:beta}\, ,
\end{IEEEeqnarray}
which is the angle between the local north $\mathbf{e}_N$ and the longitudinal symmetry axis of the wing. 

From \eqref{eqn:beta}, assuming without loss of generality $\phi_W=0$, one can see that $\beta$ converges to $\pm\pi/2$ if the wing approaches the border of the wind window, e.g $\phi\approx\pm\pi/2$.
An estimate of the wind direction $\phi_W$, needed to compute the angle $\beta$, can be either obtained by measurements provided by ground based sensors or by processing the measurements of the line force collected during the traction phase, see e.g. \cite{ZFM_TCST14}.

The considerations presented so far lead to the idea of extending the definition of the velocity angle $\gamma$ by a regularization term such that the wing's orientation is also defined for static positions of the wing.  In particular, we define the regularized velocity angle as (compare with \eqref{eqn:Gamma} and \eqref{eqn:beta}):
\begin{IEEEeqnarray}{rCl}\label{eqn:Gamma_reg}
 \gamma^{\, r} &=& \arctan{\left(\frac{\cos{(\theta)}\dot{\phi} + c\sin{(\phi-\phi_W)}}
                                      {\dot{\theta}             + c\sin{\theta}\cos{(\phi-\phi_W)}}\right)}\, ,
\end{IEEEeqnarray}
where $c>0$ is a scalar chosen by the control designer. In principle, the value of $c$ should reflect the magnitude of the absolute wind speed (divided by the tether length) which might be quite difficult to obtain. However, in simulations and experiments the system behavior resulted to be not sensitive to this quantity, due to the relatively large line length values ($50$-\Meter{200}) compared to the absolute wind speed ($3$-\mps{6}).

Thus, according to \eqref{eqn:Gamma_reg}, during the traction phase when the speed of the wing is significantly larger than the wind speed we have $\gamma^{\, r} \approx \gamma$, but during the retraction phase, when the wing speed approaches zero, $\gamma^{\, r}$ still provides a reasonable value whereas $\gamma$ of \eqref{eqn:Gamma} becomes undefined. 
A comparison between $\gamma(t)$ and $\gamma^{\, r}(t)$ during a flight test is shown in  \figref{fig:GAvsGAreg_PwrCycle_SKPdata}.
\begin{figure}[htb]
 \begin{center}
  \includegraphics[trim= 0cm 0cm 0cm 0cm,width=0.5\textwidth]{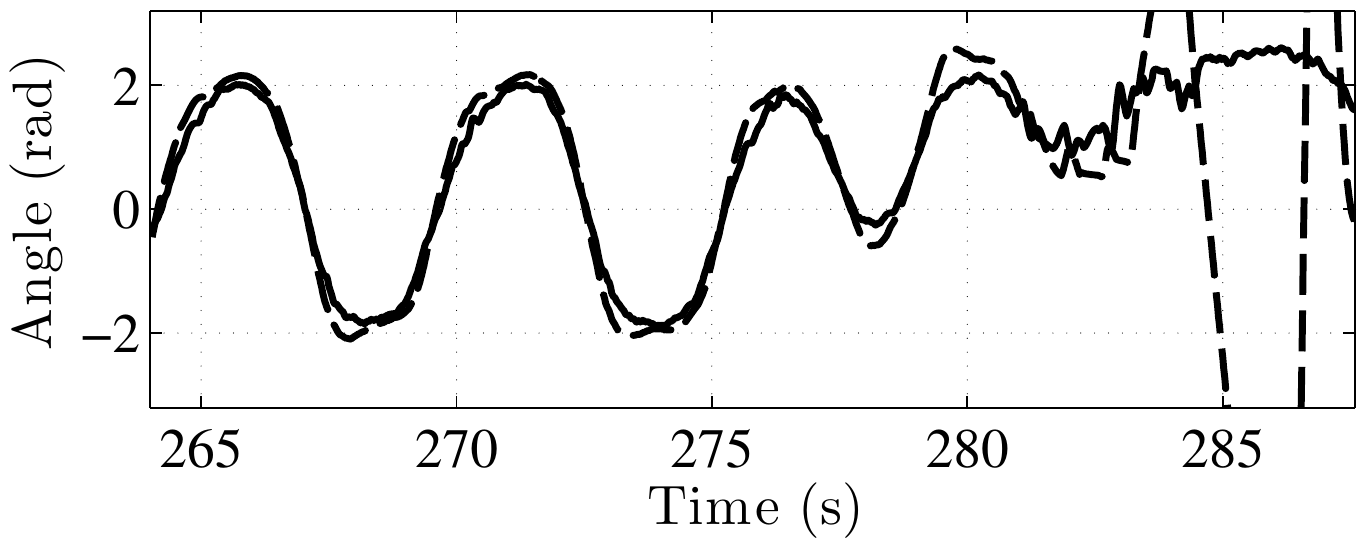}
  \caption[]{Experimental data. Time courses of $\gamma(t)$ (dashed) and $\gamma^{\, r}(t)$ (solid) during a transition from flying figure-eight paths in crosswind conditions (up to approximately $\seconds{282}$) to a position at the border of the wind window.}
  \label{fig:GAvsGAreg_PwrCycle_SKPdata}
 \end{center}
\end{figure}

With the regularized velocity angle \eqref{eqn:Gamma_reg} we can now adopt a control scheme for the retraction phase similar to the one used for the traction phase controller described in \cite{FaZg13}.


In particular, we consider a hierarchical control scheme consisting of three nested loops, shown in \figref{fig:Control_Scheme_GAreg}. Note that the regularized velocity angle cannot be directly measured and need to be estimated, see \cite{FHBK12} for details. 
\begin{figure}[htb]
 \begin{center}
  \includegraphics[trim= 0cm 0cm 0cm 0cm,width=0.5\columnwidth]{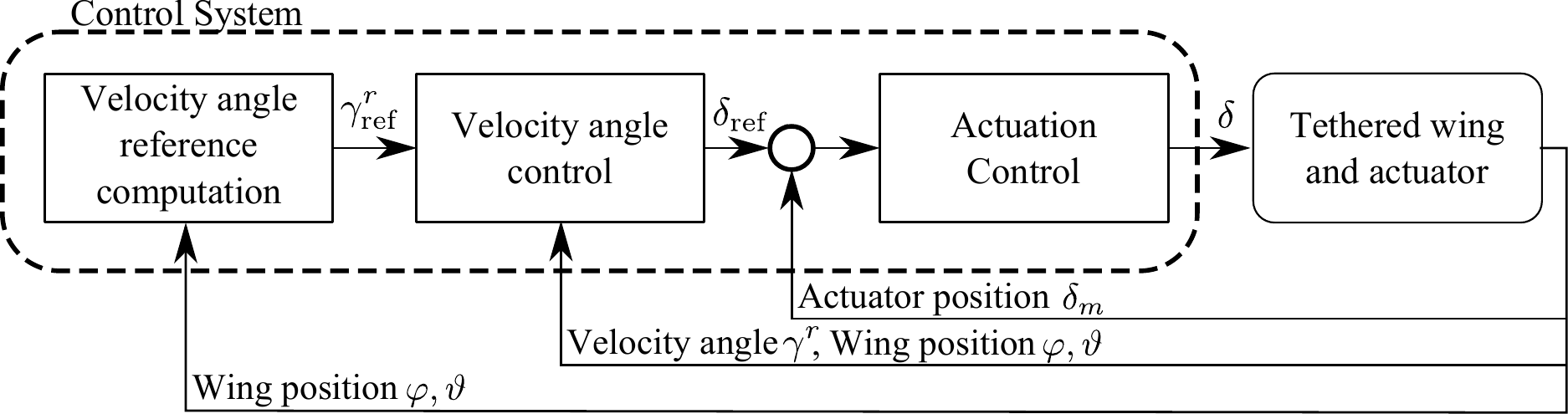}
  \caption[]{Control scheme overview using the regularized velocity angle}
  \label{fig:Control_Scheme_GAreg}
 \end{center}
\end{figure}
Besides the use of the regularized velocity angle as feedback variable the main difference between the retraction and the traction phases lies in the computation of the velocity angle reference $\gamma_{\text{ref}}^{\, r}$.
Therefore, we will only recall briefly the equations describing the inner control loops for the sake of completeness (see \cite{FaZg13} for details) and focus here on the outer control loop, responsible for providing the velocity angle controller with a suitable reference. 

Neglecting higher-order effects and external disturbances, the actuation system can be modeled as a second order system. The closed loop dynamics of the actuation control loop are then given by
\begin{IEEEeqnarray}{rCl}\label{eqn:Actuation_dyn}
 \ddot{\delta}_m &=& \omega_{\text{cl}}^2\delta_{\text{ref}}-2\zeta_{\text{cl}}\omega_{\text{cl}}\dot{\delta}_m-\omega_{\text{cl}}^2\delta_m\, ,
\end{IEEEeqnarray}
where $\delta_m$ is the actuator's position, $\delta_{\text{ref}}$ is the actuator's position reference, and $\omega_{\text{cl}}$ and $\zeta_{\text{cl}}$ are the natural frequency and damping, respectively, of the actuation control loop. The steering deviation is then obtained as $\delta=K_\delta\delta_m$, where $K_\delta$ is a known constant depending on the mechanical setup of the system. In our case, $K_\delta=1$. The velocity angle control loop consists of a proportional controller given by
\begin{IEEEeqnarray}{rCl}\label{eqn:VelocityAngleCtrlLaw}
 \delta_{\text{ref}} &=& K_c\left(\gamma^{\, r}_{\text{ref}}-\gamma^{\, r}\right)\, ,
\end{IEEEeqnarray}
where the gain $K_c$ is chosen by the designer.

As already mentioned, the goal of the retraction controller is to stabilize the wing at a static position in terms of $\phi$ and $\theta$ at the border of the wind window, e.g. $\phi-\phi_W=\pm\pi/2$, and at a given elevation angle $\theta_{\text{ref}}$. 
As seen in the previous section from \eqref{eqn:Gamma_reg}, we have $\gamma^{\, r}=\pi/2$ for a static position of the wing with $\phi-\phi_W=\pi/2$. This corresponds to a wing position on the left of the wind window as seen from the GU. Similarly, if a position on the right of the wind window is considered, i.e. $\phi-\phi_W=-\pi/2$, the regularized velocity angle becomes $\gamma^{\, r}=-\pi/2$. For simplicity, we will now only consider positions on the left of the wind window for the retraction phase, i.e. $\phi-\phi_W=\pi/2$ (the application to positions on the right of the wind direction is straightforward). 

Using the point-mass model \eqref{eqn:EqnMot} of the tethered wing, it can be shown that there exist equilibrium points at the border of the wind window, whose values are a function (for a given wing) of the steering input $\delta$ and the absolute wind speed. These equilibrium points can be computed as usual by setting all time derivatives of the model states to zero and solving \eqref{eqn:EqnMot} for a given steering input. Additionally, they can also be found by numerical simulations of the point-mass model employing a constant steering input. This suggest that these equilibrium points are open-loop stable and have a non-empty region of attraction, as it is revealed also by commonly used analysis techniques (see e.g. \cite{Khalil01}).

Inspired by the above considerations, we propose the following feedback control strategy to compute a reference value for the velocity angle:
\begin{IEEEeqnarray}{rCl}\label{eqn:GA_Ref_Retract}
 \gamma^{\, r}_{\text{ref}} = K_{\theta}(\theta_{\text{ref}}-\theta)+\frac{\pi}{2},\qquad K_{\theta}<0\,,
\end{IEEEeqnarray}
where $\theta_\text{ref}$ is a reference elevation angle chosen by the user, which should theoretically correspond to an equilibrium point for the wing at the side of the wind window.
From \eqref{eqn:GA_Ref_Retract}, one can note that, if the elevation of the wing is smaller than the reference elevation, the velocity angle reference is smaller than $\pi/2$, thus demanding the wing to move towards the zenith of the wind window, and vice-versa for a larger elevation than the reference elevation we have $\gamma^{\, r}_{\text{ref}}>\pi/2$. This reference is saturated to $\gamma^{\, r}_{\text{ref}}\in[\gamma_{\text{min}},\gamma_{\text{max}}]$ to prevent the wing from turning away too much from the wind direction. Such situation could in fact give place to a transient in crosswind conditions, which would increase the traction force unnecessarily.

The scalar gain $K_c$ for the velocity angle controller and the scalar gain $K_{\theta}$ for the reference computation are chosen by the designer.
By using \eqref{eqn:GA_Ref_Retract} in the outer loop of the control scheme (see \figref{fig:Control_Scheme_GAreg}), the resulting control system is linear (time varying) and controller gains $K_{\theta}$ and $K_c$ can be found, such that robust stability is achieved in the face of model uncertainty and different wind conditions. 
In particular, we can rewrite the system dynamics in terms of angle errors
\begin{IEEEeqnarray}{rCl}
 \Delta\gamma^{\, r}    &=& \gamma^{\, r}_{\text{ref}}-\gamma^{\, r}\\
 \Delta\theta           &=& \theta_{\text{ref}} - \theta \label{eqn:DeltaTH_def}
\end{IEEEeqnarray}
and of the position and velocity of the actuation system, $\delta_m$ and $\dot{\delta}_m$. 
In order to formulate the error dynamics, we need an intermediate step to include the dynamics of the angle $\theta$. To this end, we note that the apparent wind velocity component in the tangent plane, $|\vec{W}_a^p|$, in $\theta$ direction, given by $r\dot{\theta}$, is by definition of $\gamma$ equal to $|\vec{W}_a^p|\cos{\gamma}$ (compare \eqref{eqn:Gamma_definition}) which can equivalently be written as $|\vec{W}_a^p|\sin{(\pi/2-\gamma)}$. Since the wing tends to align itself with the wind direction, we assume that $\pi/2-\gamma^{\, r}$ is small, so that we can linearize its trigonometric functions. Then, the dynamics of the $\theta$ angle can be written as:
\begin{IEEEeqnarray}{rCl}\label{eqn:TH_dyn_smpl}
 \dot{\theta} &=& \frac{|\vec{W}_a^p|}{r}\left(\frac{\pi}{2}-\gamma^{\, r}\right)\, .
\end{IEEEeqnarray}

We can now state the system dynamics by using \eqref{eqn:GammaDotLaw},\eqref{eqn:Actuation_dyn}-\eqref{eqn:DeltaTH_def}, and setting $\mathbf{x} = [\Delta\theta,\Delta\gamma^{\, r},\delta_m,\dot{\delta}_m]^T$ (where ${}^T$ stands for the matrix transpose operation) as

\begin{IEEEeqnarray}{rCl}\label{eqn:Lin_CL_Sys}
 \renewcommand\arraystretch{1.2}
 \dot{\vec{x}} &=& 
 \underbrace{\begin{bmatrix}
              K_{\theta}\frac{|\vec{W}_a^p|}{r}   & -\frac{|\vec{W}_a^p|}{r}           &  0                    & 0\\
              K_{\theta}^2\frac{|\vec{W}_a^p|}{r} & -K_{\theta}\frac{|\vec{W}_a^p|}{r} & -KK_{\delta}          & 0\\
              0                                   & 0                                  &  0                    & 1\\
              0                                   & K_c\omega_{\text{cl}}^2            & -\omega_{\text{cl}}^2 & -2\zeta_{\text{cl}}\omega_{\text{cl}}
             \end{bmatrix}
 }_{A_{\text{cl}}}
 \mathbf{x} + w\, .
\end{IEEEeqnarray}

In \eqref{eqn:Lin_CL_Sys}, the term $K$ corresponds to the uncertain gain in \eqref{eqn:GammaDotLawCoeffs1} and depends on the system's parameters as well as the wind and the flight conditions. The term $w\in\mathbb{R}^4$ accounts for effects of gravity and apparent forces of \eqref{eqn:GammaDotLawCoeffs2},
as well as for the forces exerted by the lines on the actuator. System \eqref{eqn:Lin_CL_Sys} has time-varying, uncertain linear dynamics characterized by the matrix $A_{\text{cl}}(\Theta)$, where $\Theta=[K,|\vec{W}_a^p|]$. Upper and lower bounds for such parameters can easily be derived on the basis of the available knowledge on the system. These bounds can be employed to compute points $\Theta^i,\, i=1,\ldots,n_v$, such that $\Theta\in \text{conv}(\Theta^i)$, where conv denotes the convex hull. Then, the closed-loop system \eqref{eqn:Lin_CL_Sys} results to be robustly stable if there exists a positive definite matrix $P=P^T\in\mathbb{R}^{4x4}$ such that (see e.g. \cite{Amato06}):
\begin{IEEEeqnarray}{rCl}\label{eqn:QuadStabCrit}
 A_{\text{cl}}^T(\Theta^i)P + PA_{\text{cl}}(\Theta^i) \prec 0,\, i=1,\ldots,n_v\, ,
\end{IEEEeqnarray}

Condition \eqref{eqn:QuadStabCrit} can be checked by using an LMI solver. 
In \secref{sec:results} we show with simulations and experiments that indeed a single pair $(K_c,K_\theta)$ achieves robust stability of the control system, as predicted by the described theoretical analysis. The two scalar gains, i.e. the values of $K_c$ and $K_\theta$, can be tuned at first by using the equations \eqref{eqn:GammaDotLaw} and \eqref{eqn:GA_Ref_Retract}, and then via experiments.

As shown in \secref{sec:results}, this approach is able to stabilize the wing at the border of the wind window but is dependent on an estimate of the wind speed and wind direction at the wing's location.
Since these might not be straightforward to obtain, an alternative approach is presented in the next section, which relies only on directly measurable quantities.

\subsection{Retraction Control Based on Elevation Dynamics}
\label{ssec:Elevation_Dynamics_Based_Retraction}
As an alternative to the regularized velocity angle, we propose here to use the elevation angle $\theta$ as feedback variable. The main advantage of such an approach is a higher reliability, since the elevation is directly measured (see e.g. \cite{FHBK12}) and there is no need to estimate the wind direction at the wing's location.
We will carry out the controller's design on the basis of a new model that links the elevation dynamics to the steering input, which we derive next.

From \eqref{eqn:EqnMot}, we can write the $\theta$-dynamics as:
\begin{IEEEeqnarray}{rCl}\label{eqn:THdd_dyn}
 \ddot{\theta} &=& \frac{\mathbf{F}\cdot\mathbf{e}_N}{rm}-\sin{(\theta)}\cos{(\theta)}\dot{\phi}^2-\frac{2}{r}\dot{\theta}\dot{r}\, .
\end{IEEEeqnarray}

We consider the following assumptions:
\begin{assumption}
 (Steady State) The wing is at a steady state angular position at the border of the wind window. \hfill$\blacksquare$
 \label{as:SteadyState}
\end{assumption}
\begin{assumption}
 (Small roll angle) The control input $\psi$ is sufficiently small such that its trigonometric functions can be linearized.\hfill$\blacksquare$
 \label{as:SmallInput}
\end{assumption}

\asref{as:SteadyState} implies that the sum of the forces acting on the wing in $\mathbf{e}_E$ direction is zero and that the angular velocities of the wing are small. Thus, effects from apparent forces are small. Moreover, this also implies that the wing's longitudinal axis is aligned with the apparent wind direction.
\asref{as:SmallInput} is also reasonable, since for example during our test flights the roll angle was within $\pm\degr{18}$.
We can now state our result concerned with the wing's model:

\begin{proposition}
 \label{prop:Elevation_Dynamics}
 Let assumptions \ref{as:SteadyState}-\ref{as:SmallInput} hold. Then, the $\theta$ dynamics \eqref{eqn:THdd_dyn} can be written as
 \begin{IEEEeqnarray}{rCl}\label{eqn:TH_dyn_approx}
 \ddot{\theta} &=& -\mathcal{C}\delta-\frac{g\cos{(\theta)}+2\dot{\theta}\dot{r}}{r}\, ,
\end{IEEEeqnarray}
where
\begin{IEEEeqnarray}{rCl}\label{eqn:C_gain}
 \mathcal{C}   &=& \frac{\rho AC_L}{2rmd_s}\left(1+\frac{1}{E_{eq}^2}\right)W_0\sin{(\phi)}|\vec{W}_a|\, .
\end{IEEEeqnarray}
\end{proposition}

\begin{proof}
 See the Appendix.\hfill$\blacksquare$ 
\end{proof}

The model in \eqref{eqn:TH_dyn_approx} gives a direct relationship between the input $\delta$ and the elevation of the wing $\theta$. 
It is worth elaborating a bit more on this result and its implications.
As we can see from \eqref{eqn:TH_dyn_approx}, gravity and apparent forces have less influence with increasing tether length, since the linear acceleration remains constant such that the angular one is inversely proportional to the radius.
The term $\rho A C_L/(2rmd_s)(1+1/E_{eq}^2)$ in \eqref{eqn:C_gain} remains roughly constant during the retraction and is specific to the employed wing.
Equation \eqref{eqn:CWTFmodel} also implies that a larger area-to-mass ratio, $A/m$, gives in general a higher gain $\mathcal{C}$, and that the steering gain of wings with similar aerodynamic coefficients but different sizes should not change much, provided that they have similar $A/m$.

As regards the design of the retraction controller exploiting the model \eqref{eqn:TH_dyn_approx}, we consider again a hierarchical control loop, now consisting only of two nested loops, the actuation control loop and the elevation controller, shown in \figref{fig:Control_Scheme_THpos}.
\begin{figure}[htb]
 \begin{center}
  \includegraphics[trim= 0cm 0cm 0cm 0cm,width=0.5\columnwidth]{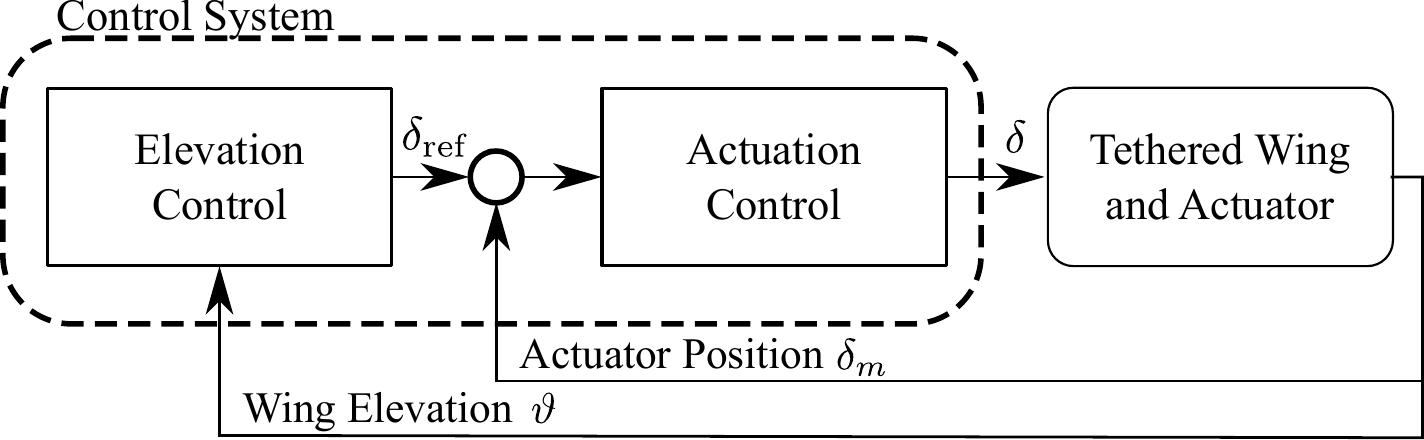}
  \caption[]{Control scheme overview using the wing elevation based controller}
  \label{fig:Control_Scheme_THpos}
 \end{center}
\end{figure}

To design the elevation controller, \eqref{eqn:TH_dyn_approx} is first linearized around an equilibrium point, which serves as reference position $\theta_{\text{ref}}$. As pointed out in \secref{ssec:Velocity_Angle_Based_Retraction}, such an equilibrium point can be found using the point-mass model \eqref{eqn:EqnMot}.
The resulting linear system is given by
\begin{IEEEeqnarray}{rCl}\label{eqn:UncertainLinSystem}
 \dot{\vec{x}}' &=& \begin{bmatrix}
                         0 & 1\\ \frac{g\sin{(\theta_{\text{ref}})}}{r} & -\frac{2\dot{r}}{r}
                       \end{bmatrix}\,\vec{x}'+
                       \begin{bmatrix}
                         0\\ -\mathcal{C}
                       \end{bmatrix}\,u\, ,
\end{IEEEeqnarray}
where $\vec{x}'=[\Delta\theta,\Delta\dot{\theta}]^T$ and $u=\delta$. The tracking error in $\theta$ and $\dot{\theta}$ are given as
\begin{IEEEeqnarray}{rCl}
 \Delta\theta       &=& \theta_{\text{ref}}-\theta\\
 \Delta\dot{\theta} &=& \dot{\theta}_{\text{ref}}-\dot{\theta}\, ,\label{eqn:DeltaDotTH}
\end{IEEEeqnarray}
where the reference values correspond to a static angular position, i.e. $\dot{\theta}_{\text{ref}}=0$.

We use a state feedback controller $K_\text{SF}$ of the form
\begin{IEEEeqnarray}{rCl}\label{eqn:K_statefeedback}
 z &=& -K_{\text{SF}}\,\vec{x}'\, ,
\end{IEEEeqnarray}
where $z=\delta_\text{ref}$ and $K_{\text{SF}}= [k_1^\text{SF}\; k_2^\text{SF}]$ is a vector of feedback gains that can be designed by means of standard techniques like pole placement ore linear-quadratic (LQ) regulation.
Again, it can be shown that there exists a matrix $K_{\text{SF}}$ for which the system is robustly stabilized in the presence of the uncertain, time-varying parameters. 
A robustness analysis can be carried out similarly to the one in \secref{ssec:Velocity_Angle_Based_Retraction}; the corresponding closed loop dynamics are given, using 
\eqref{eqn:Actuation_dyn}, \eqref{eqn:UncertainLinSystem}-\eqref{eqn:K_statefeedback}, $\delta=K_\delta\delta_m$, and $\vec{x}''= [\Delta\theta,\Delta\dot{\theta},\delta_m,\dot{\delta}_m]$, by
\begin{IEEEeqnarray}{rCl}\label{eqn:Lin_CL_Sys_Pos}
 \renewcommand\arraystretch{1.2}
 \dot{\vec{x}}'' = 
 \underbrace{\begin{bmatrix}
              0                              & 1                             &  0                                              & 0\\
              \frac{g\sin{(\theta_\text{ref})}}{r}      & -\frac{2\dot{r}}{r}           & -\mathcal{C}K_\delta & 0\\
              0                              & 0                             &  0                                              & 1\\
              -\omega^2_{\text{cl}}k_1^{SF}  & -\omega^2_{\text{cl}}k_2^{SF} & -\omega_{\text{cl}}^2                           & -2\zeta_\text{cl}\omega_{\text{cl}}
             \end{bmatrix}
 }_{A_{\text{cl}}}
 \vec{x}''+ w\, .
\end{IEEEeqnarray}
Here, the uncertain time-varying parameters are given by $\Theta = [r,\dot{r},\mathcal{C}]$.


\subsection{Discussion}
\label{ssec:RetrCtrl_Discussion}
We presented two control approaches for the retraction phase, one based on a regularized version of the velocity angle $\gamma$ and one based on the $\theta$-dynamics derived from the first principle model \eqref{eqn:EqnMot}. In the latter, we exploit a direct link between the input $\delta$ and the angular acceleration $\ddot{\theta}$, while the first approach does not directly consider explicitly the $\theta$ dynamics and relies on the turning rate $\dot{\gamma}$ instead. For the sake of comparison, also in the first approach one can extract the $\theta$ dynamics, in particular by considering \eqref{eqn:TH_dyn_smpl}, i.e. 
\begin{IEEEeqnarray*}{rCl}
 \dot{\theta} &=& \frac{|\vec{W}_a^p|}{r}\left(\frac{\pi}{2}-\gamma^{\, r}\right)\, ,
\end{IEEEeqnarray*}
where $|\vec{W}_a^p|$ is the apparent wind velocity projected onto the tangent plane to the wind window at the wing's location.
By taking the time derivative of \eqref{eqn:TH_dyn_smpl} and combining it with \eqref{eqn:GammaDotLaw}, and assuming that the apparent wind velocity $\vec{W}_a$ is constant and again that the wing stays at a constant $\phi$ position, we obtain:
\begin{IEEEeqnarray}{rCl}
 \label{eqn:THdd_dyn_smpl}
 \ddot{\theta} = -\frac{\rho AC_L}{2rmd_s}\left(1+\frac{1}{E_{eq}^2}\right)|\vec{W}_a|^2\delta - \frac{g\cos{(\theta)}}{r}\,.
\end{IEEEeqnarray}

Comparing this equation with the one derived from the model \eqref{eqn:TH_dyn_approx}, one can see a few differences. First, the second term in the right-hand side of \eqref{eqn:THdd_dyn_smpl} does not contain the term related to apparent forces. This comes from the fact that the $\theta$-dynamics in \eqref{eqn:TH_dyn_smpl} do not consider the influence of the reeling speed $\dot{r}$.
The term related to gravity is the same since we assume a $\gamma\approx\pi/2$ for the retraction. Note that, as one would expect, for both models the influence of the additive terms on the angular acceleration become smaller for longer tether length.

The gain relating the input $\delta$ to $\ddot{\theta}$, denoted by $\mathcal{C}$ in \eqref{eqn:TH_dyn_approx}, is quite similar to the corresponding gain in \eqref{eqn:THdd_dyn_smpl}.
The difference comes from how the force component in $\theta$ direction, $\mathbf{F}\cdot\mathbf{e}_N$, is calculated. In \eqref{eqn:TH_dyn_approx}, this component is calculated by considering the apparent wind in the tangent plane at the wing's position, i.e. $W_0\sin{(\phi)}+r\cos{(\theta)}\dot{\phi}$ where $W_0\sin{(\phi)}$ is the dominating factor, see the Appendix.
On the other hand, the corresponding term in \eqref{eqn:TH_dyn_smpl} is $\vec{W}_a^p$ which corresponds, assuming a static angular position at the border of the wind window and constant line length, to $\vec{W}_a\simeq W_0\sin{(\phi)}$.
In summary, it has to be noted that the structure of the two models is the same, which explains why the corresponding controllers have similar qualitative behavior, as it will be shown in \secref{sec:results}, but with quite marked differences in tracking performance.

\subsection{Reeling}
\label{ssec:Retr_Reeling}
As mentioned above, the reeling can be considered, from the point of view of the position control system, as an external disturbance since its main influence is on the magnitude of the apparent wind speed and all other effects are comparably small.
This is the reason why both the traction and retraction controllers can be designed independently from the reeling speed control. 
For simplicity, we therefore adopt a simple reeling control scheme for both phases by setting a torque reference on the generators.
During the traction phase the torque reference is chosen with a feedback strategy that the optimal real-out speed \cite{Loyd80} is tracked.
In particular, assuming a steady state reeling, i.e. constant speed, where the optimal traction force has to be matched by the motor torque, we have
\begin{IEEEeqnarray}{rCl}
 T_m &=& F_c^*r_d
\end{IEEEeqnarray}
where $T_m$ is the torque applied by the motor, $r_d$ is the radius of the drum, and $F_c^*$ is the optimal traction force for maximum power production for a given wind situation.
A simplified model of the traction force $F_c$ has been first introduced in \cite{Loyd80} and then subsequently refined in several contributions, see e.g. \cite{FaMP11}:
\begin{IEEEeqnarray}{rCl}\label{eqn:CWTFmodel}
 F_c(t) = |\vec{F}_c(t)| &=& \mathscr{C}|\vec{W}_a^r|^2
\end{IEEEeqnarray}
with
\begin{IEEEeqnarray}{rCl}
 \mathscr{C} &=& \frac{1}{2}\rho AC_LE_{eq}^2\left(1+\frac{1}{E_{eq}^2}\right)^{\frac{3}{2}}\, ,
\end{IEEEeqnarray}
where $\rho$ is the air density, $A$ is the wing reference area, $C_L$ is the lift coefficient, $E_{eq}$ is the equivalent efficiency, and $\vec{W}_a^r$ is the apparent wind vector component in tether direction consisting of the wind speed $\vec{W}$ and the reeling speed $\dot{r}$. It can be shown that the optimal reeling speed is equal to one third of the wind speed in tether direction, see e.g. \cite{Loyd80}.
Therefore we can expresse $\vec{W}_a^r$ as
\begin{IEEEeqnarray}{rCl}
 |\vec{W}_a^r| &=& W^r-\dot{r}\\
               &=& 3\dot{r}^*-\dot{r}\, ,
\end{IEEEeqnarray}
where $W^r$ is the wind speed in tether direction. Thus the motor torque to achieve a desired reel-out speed $\dot{r}$ is given as
\begin{IEEEeqnarray}{rCl}
 T_m &=& \mathscr{C}r_d\left(3\dot{r}^*-\dot{r}\right)^2\, .
\end{IEEEeqnarray}
Then, setting the motor torque equal to
\begin{IEEEeqnarray}{rCl}\label{eqn:T_mot_opt}
 T_m &=& 4\mathscr{C}r_d\dot{\hat{r}}^2\, ,
\end{IEEEeqnarray}
where $\dot{\hat{r}}$ is the actual measured reeling speed, leads to a steady state solution of $\dot{\hat{r}}=\dot{r}^*$.
It can be shown that such a solution is an asymptotically stable steady-state when the feedback reeling strategy \eqref{eqn:T_mot_opt} is used.

Additionally, we included a lower and an upper bound on the torque reference to avoid wing stall and mechanical overload of the system, respectively. 
During the retraction, a constant torque reference is chosen to achieve a high reel-in speed, in order to increase the duty-cycle of the overall power generation scheme.

Indeed, the interplay between the wing dynamics and reeling speed could be exploited using a multivariable control technique with the aim to optimize the power output.
If an additional actuator to change the pitch angle of the wing is also present, i.e. allowing one to change the lift and drag coefficients of the wing, the efficiency of the system could be further increased. 
These topics are not considered in this paper but they represent further research directions.

%% file: results.tex
\section{Results}
\label{sec:results}

We first compare the proposed control approaches for the retraction phase in simulation, employing the non-linear point-mass model for tethered wings \eqref{eqn:EqnMot}.
The main system and controller parameters are shown in Table \ref{tab:system_params} and Table \ref{tab:control_params_Ctrl}, respectively. The terms relating to $\gamma^{\, r}$ apply only to the approach from \secref{ssec:Velocity_Angle_Based_Retraction} and for the state feedback approach an LQ regulator with weighting matrices equal to the identity matrix were used.

\ctable[caption = System Parameters,
        cap     = System Parameters,
        label   = tab:system_params,
        pos     = tbp,
        maxwidth = \linewidth,
        captionskip = -2ex
       ]{lcrl}{
}{ \FL
Name & Symbol & Value & Unit\NN
\cmidrule(r){1-4}
Wing effective area      & $A$      & $9$           & \si{\meter\squared} \NN
Kite span                & $d_s$    & $2.7$         & \si{\meter}\NN
Kite mass                & $m$      & $2.45$        & \si{\kilogram} \NN
Tether length            & $r$      & $[50\ldots 150]$ & \si{\meter} \NN
Tether diameter          & $d_t$    & $0.003$       & \si{\meter} \NN
Tether density           & $\rho_t$ & $970$         & \si[per-mode=symbol]{\kilogram\per\cubic\meter} \NN
Air density              & $\rho$   & $1.2$         & \si[per-mode=symbol]{\kilogram\per\cubic\meter} \LL
}

\ctable[caption = {Control Parameters},
        cap     = {Control Parameters},	
        label   = {tab:control_params_Ctrl},
        pos     = tbp,
        maxwidth = \linewidth,
        captionskip = -2ex,
       ]{lcrl}{
}{\FL       
Name & Symbol & Value & Unit\NN
\cmidrule(r){1-4}
Actuator control loop damping                            & $\zeta_\text{cl}$       & $0.7$          & $-$ \NN
Actuator control loop natural frequency                  & $\omega_\text{cl}$      & $78$           & \si[per-mode=symbol]{\radian\per\second} \NN
Mechanical actuation ratio                               & $K_\delta$              & $1$            & $-$ \NN
$\gamma^{\, r}$ feedback gain (traction)                 & $K_{\text{c}}$          & $0.056$        & \si[per-mode=symbol]{\meter\per\radian} \NN
$\gamma^{\, r}$ feedback gain (retraction)               & $K_{\text{c}}$          & $0.28$         & \si[per-mode=symbol]{\meter\per\radian} \NN
$\gamma^{\, r}_{\text{ref}}$  feedback gain (retraction) & $K_{\theta}$            & $-2.5$         & $-$ \NN
State feedback control gain 1 (retraction)               & $k^{\text{SF}}_1$       & $-1.4$         & \si[per-mode=symbol]{\meter\per\radian} \NN 
State feedback control gain 2 (retraction)               & $k^{\text{SF}}_2$       & $-4.6$         & \si[per-mode=symbol]{\meter\second\per\radian} \NN 
Elevation reference (retraction)                         & $\theta_{\text{ref}}$   & $1$            & \si{\radian}
\LL
}

In \figref{fig:3DTraj_PwrCycle}, a typical trajectory of the wing from launch until the end of the first  power cycle is shown. At first, the wing is flown in crosswind conditions, flying figure-eight paths until it reaches the maximum tether length of $\Meter{150}$, using the controller described in \cite{FaZg13}. Then, the retraction phase is started using either the controller based on the regularized velocity angle \eqref{eqn:VelocityAngleCtrlLaw}-\eqref{eqn:GA_Ref_Retract} or the feedback controller \eqref{eqn:K_statefeedback}, while the tether is reeled-in until a length of $\Meter{50}$ is reached. 
At that point, the traction phase controller of \cite{FaZg13} is used again to complete the power cycle. In \figref{fig:PHTH_PwrCycle}, the time courses of the position angles $\phi$ and $\theta$ during one power cycle for both control approaches are shown. Around $\seconds{73}$, the controller switches from traction to retraction and tracks the reference $\theta_{\text{ref}}=\radi{1}$. Note that $\phi$ becomes slightly larger than $\pi/2$ due to the reel-in speed, indicating that the wing surpasses the GU location against the wind. Around $\seconds{138}$, the controller switches from retraction to traction and the wing starts again flying figure-eight paths in crosswind conditions.

\begin{figure}[tbh]
 \begin{center}
    \includegraphics[width=0.5\columnwidth]{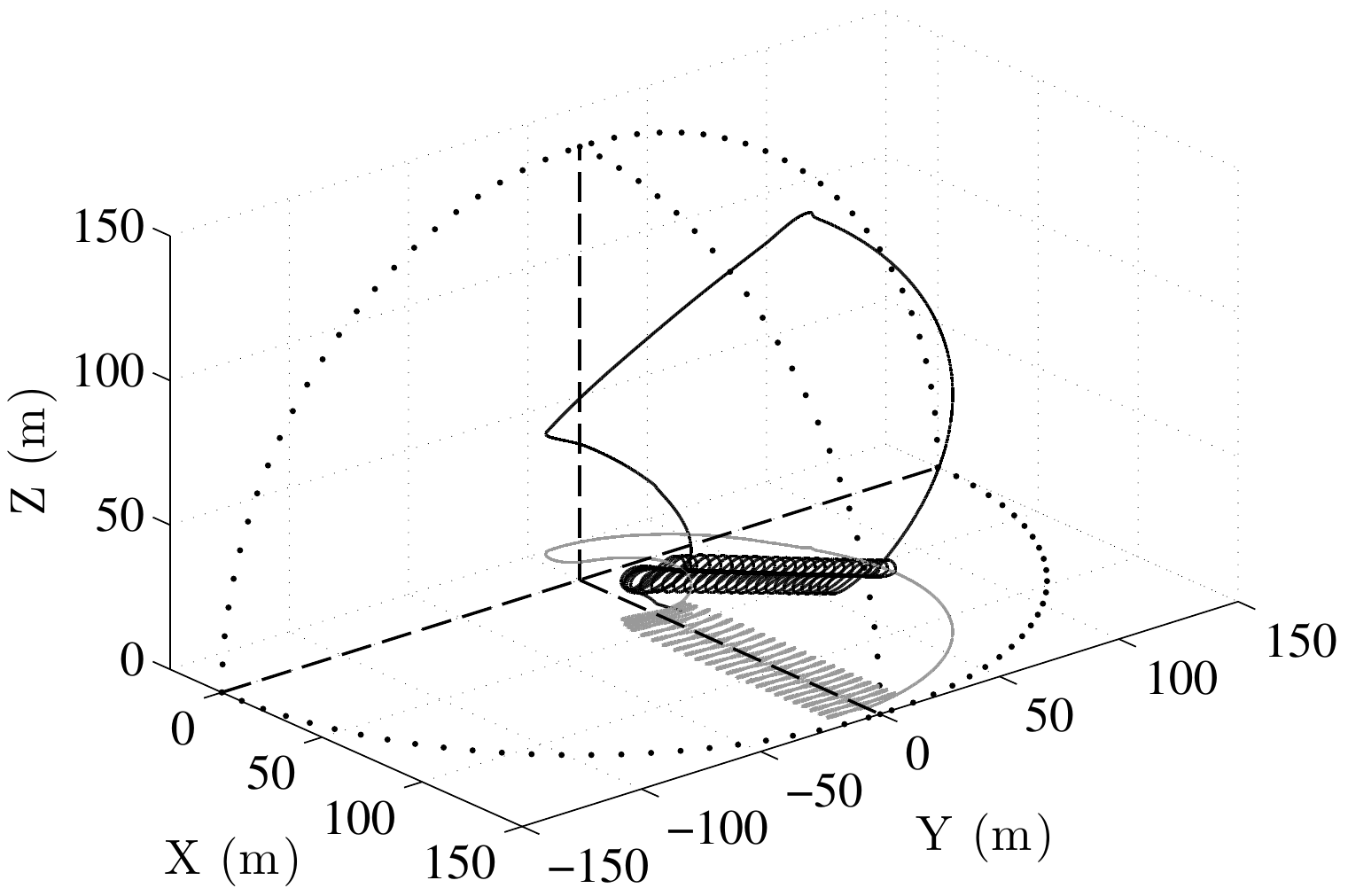}
    \caption{Simulation results. Typical 3D trajectory (black) and its projection (gray) on the ground of the tethered wing during one flown power cycle.}
    \label{fig:3DTraj_PwrCycle}
 \end{center}
\end{figure}

\begin{figure}[t]
  \centering
  \begin{subfigure}[b]{0.5\columnwidth}
   \begin{center}
    \includegraphics[width=\columnwidth]{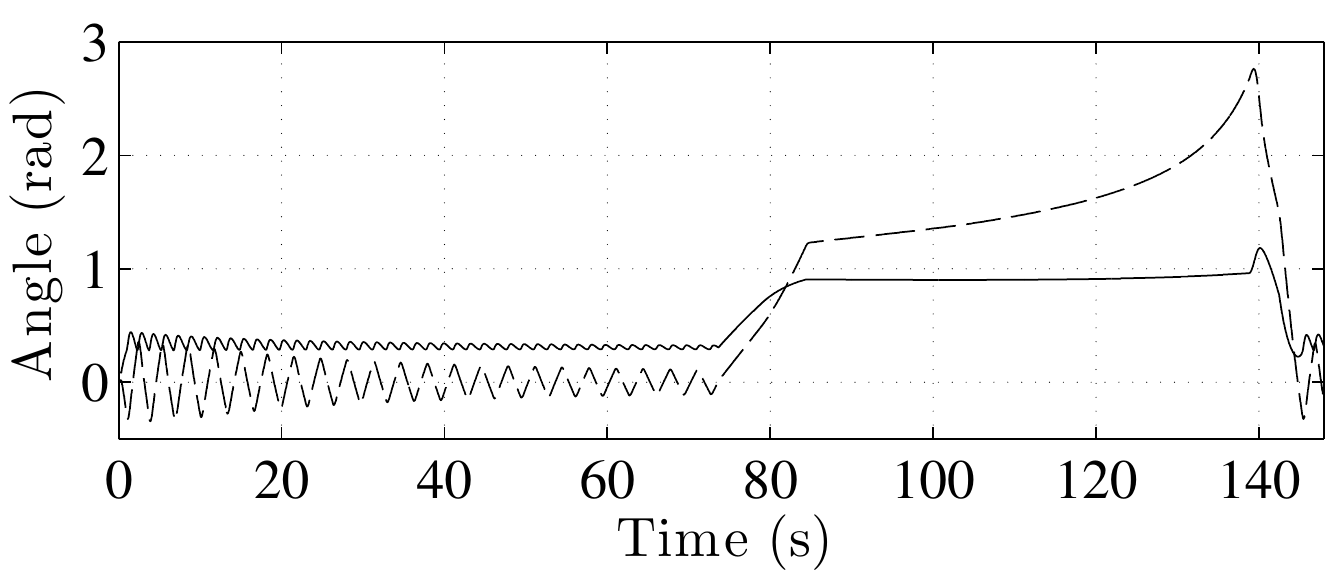}
    \caption{Using the regularized velocity angle for the retraction.}
    \label{fig:PHTH_PwrCycle_RegGACtrl}
   \end{center}
  \end{subfigure}
  \\
  \vspace{0.5cm}
  \begin{subfigure}[b]{0.5\columnwidth}
   \begin{center}
    \includegraphics[width=\columnwidth]{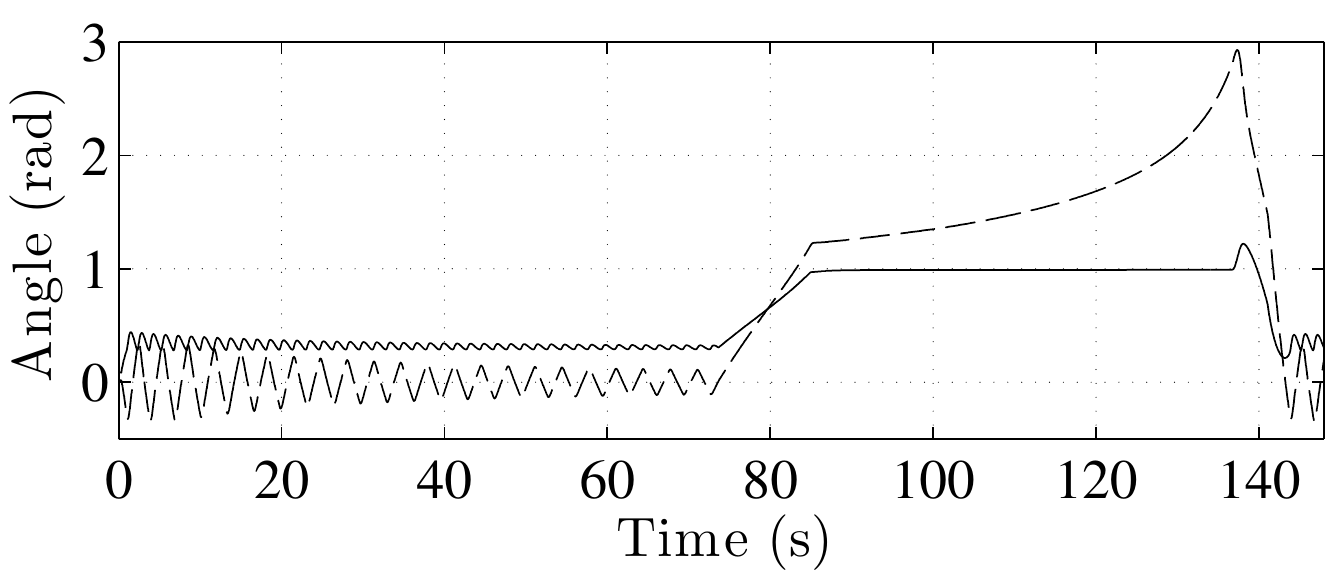}
    \caption{Using a state feedback controller for the retraction.}
    \label{fig:PHTH_PwrCycle_PosCtrl}
   \end{center}
  \end{subfigure}
  \caption{Simulation results. Time courses of $\phi$ (dashed) and $\theta$ (solid) of one power cycle with a reel-in speed of $\mps{2.5}$ and $W_0=\mps{5}$.}
  \label{fig:PHTH_PwrCycle}
\end{figure}

Both control approaches lead to qualitatively similar results, as it can be seen from \figref{fig:PHTH_PwrCycle}. The main noticeable difference is the tracking of the $\theta$ reference during retraction which is better achieved by the approach using a state feedback controller. This is expected, since the latter controller employs directly the elevation angle and its rate, which are both measured with good accuracy, as feedback variables, while the former controller uses the elevation angle to compute a reference for the regularized velocity angle, whose estimate can be inaccurate due to the uncertainty in the wind speed estimation (i.e. the tuning parameter $c$ in \eqref{eqn:Gamma_reg}). Such uncertainty gives rise to a bias in the feedback variable, which in turn reflects into a larger tracking error.
This is shown in \figref{fig:TrackErr_PwrCycles} where the average tracking error of one retraction phase for different reel-in speeds and different wind speeds, respectively, are plotted. 

 \begin{figure}[t]
  \centering
  \begin{subfigure}[b]{0.5\columnwidth}
   \begin{center}
    \includegraphics[trim= 0cm 0cm 0cm 0cm,width=\columnwidth]{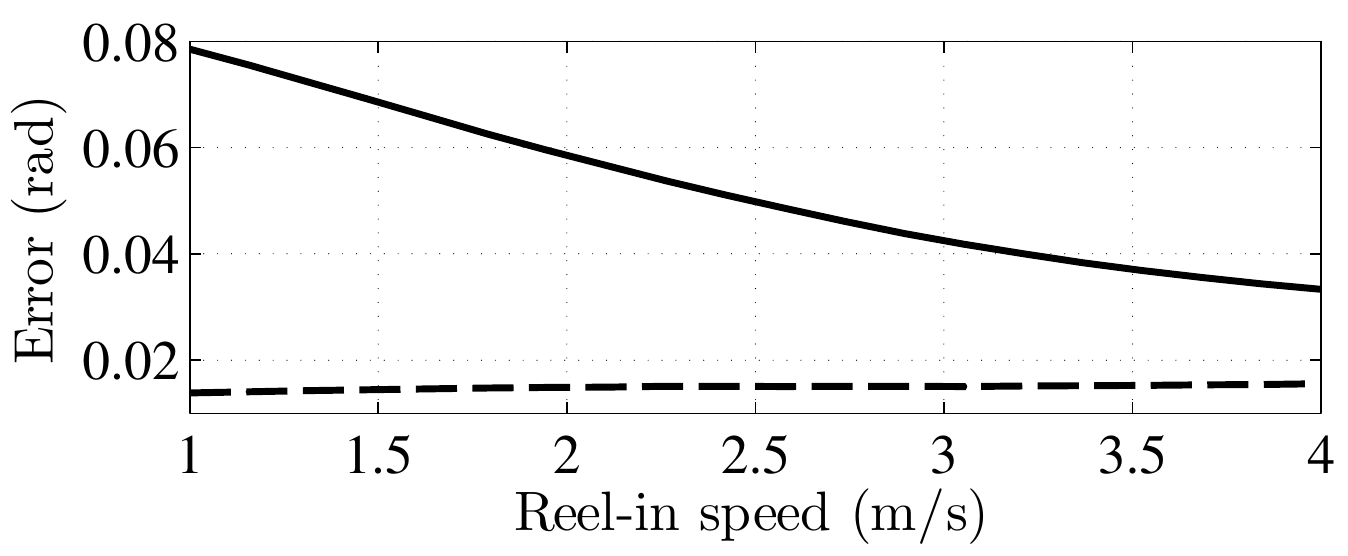}
    \caption[]{Average $\theta$ tracking error for different reel-in speeds for the regularized velocity angle based control (solid) and for the state feedback control (dashed) during one retraction phase with $W_0=\mps{5}$.}
    \label{fig:TrackErr_PwrCycles_Rdot}
   \end{center} 
  \end{subfigure}
  \\
  \vspace{0.5cm}
  \begin{subfigure}[t]{0.5\columnwidth}
   \begin{center}
    \includegraphics[trim= 0cm 0cm 0cm 0cm,width=\columnwidth]{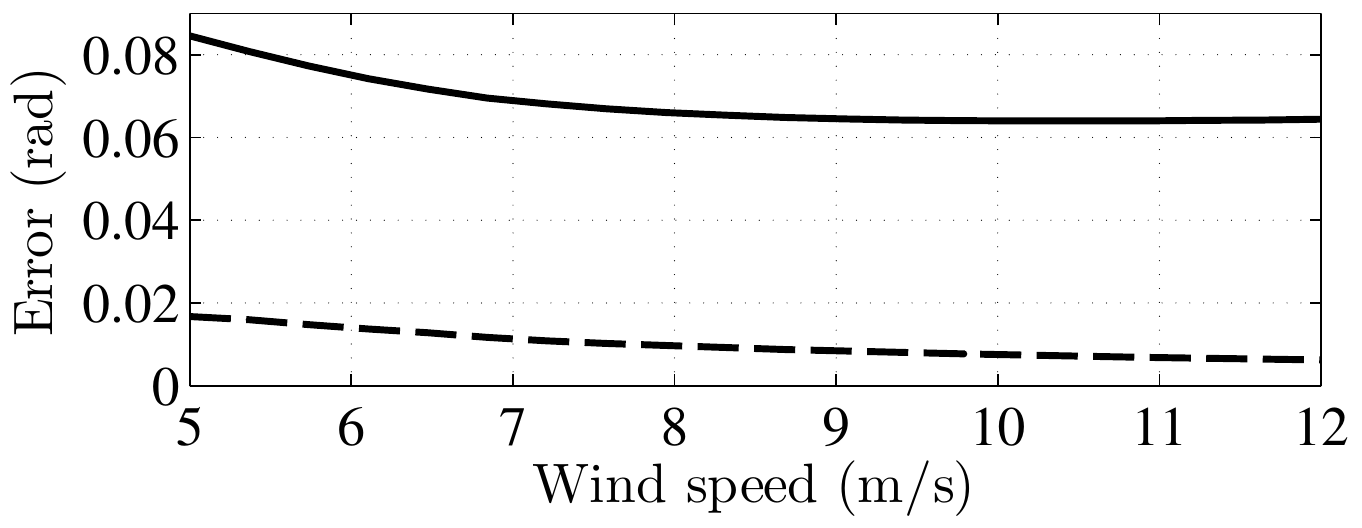}
    \caption[]{Average $\theta$ tracking error for different wind speeds $W_0$ using the regularized velocity angle based control (solid) and the state feedback control (dashed) during one retraction phase with a reel-in speed of $\mps{2.5}$.}
    \label{fig:TrackErr_PwrCycles_W0}
   \end{center}
  \end{subfigure}
  \caption{Average $\theta$ tracking error during one retraction phase.}
  \label{fig:TrackErr_PwrCycles}
\end{figure}


Real-world experiments have also successfully been carried out on the Swiss Kite Power prototype, shown in \figref{fig:GS1_front}. The employed wing was a three line \airush surf kite with an area of $\sqm{9}$.
In \figref{fig:Exp_PwrCycle_GAreg}, the results of experimental test flights employing the retraction control strategy proposed in \secref{ssec:Velocity_Angle_Based_Retraction} using the regularized velocity angle is shown. In \figref{fig:Exp_PHTHTraj_PwrCycle_GAreg}, the wing path during a power cycle in the $(\phi,\theta)$-plane is shown. The wing is controlled to fly along figure-eight paths until it reaches the maximum tether length and then flies horizontally to the border of the wind window. Such a transient phase can be achieved by setting a new target point for the traction controller at the border of the wind window. Then, the retraction controller stabilizes the wing during the reel-in of the tether. Once at the minimum tether length, the wing turns back to fly figure-eight paths roughly aligned with the wind direction. In \figref{fig:Exp_Gamma_PwrCycle_GAreg}, the velocity angle and its reference are shown, and in \figref{fig:Exp_PHTH_PwrCycle_GAreg} the corresponding time courses of $\phi$ and $\theta$ are shown.
Note that the wing flies downwards to a low $\theta$ angle when starting a new traction phase. This is due to the increasing wing speed and rather small $K_c$ gain used for this maneuver.
This problem can be alleviated by increasing the steering gain for this phase, as we show later in \figref{fig:Exp_PHTHTraj_PwrCycle_ElCtrl}. 
A projection of the wing path on the ground plane can be seen in \figref{fig:Exp_3DTrajProjXY_PwrCycle_GAreg}. Note that the wing surpasses the GU upwind, since it reaches a negative position in the $\vec{e}_x$ direction.
The average wind speed was approximately \mps{4.6}. The time course of the wind measured roughly \Meter{5} above the ground can be seen in \figref{fig:Exp_Wspd_PwrCycle_GAreg}. The resulting traction force on the main line during the power cycle is shown in \figref{fig:Exp_Fmain_PwrCycle_GAreg}.
It can be seen that there is a significant drop in traction force during the retraction phase as expected from the considerations above, leading to a positive net energy output of the system.
The time course of the tether length can be seen in \figref{fig:Exp_R_PwrCycle_GAreg}. A movie of the autonomous power cycles is available online: \cite{Wing_moviePwrCycle}.

\begin{figure}[t]
 \centering
 \begin{subfigure}[b]{0.49\columnwidth}
  \begin{center}
   \includegraphics[trim= 0cm 0cm 0cm 0cm,width=\columnwidth]{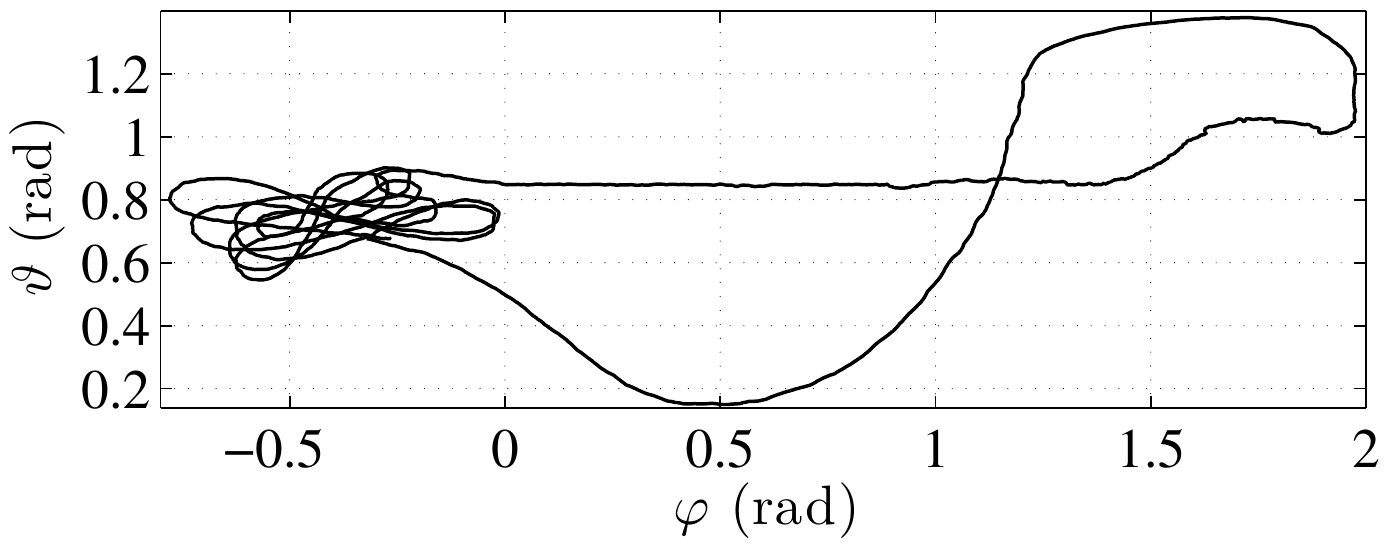}
   \caption[]{The wing trajectory in $\phi$ and $\theta$. The wind direction was roughly $\phi_W\approx\radi{-0.4}$.}
   \label{fig:Exp_PHTHTraj_PwrCycle_GAreg}
  \end{center}
 \end{subfigure}
 \hfill
 \begin{subfigure}[b]{0.49\columnwidth}
  \begin{center}
   \includegraphics[trim= 0cm 0cm 0cm 0cm,width=\columnwidth]{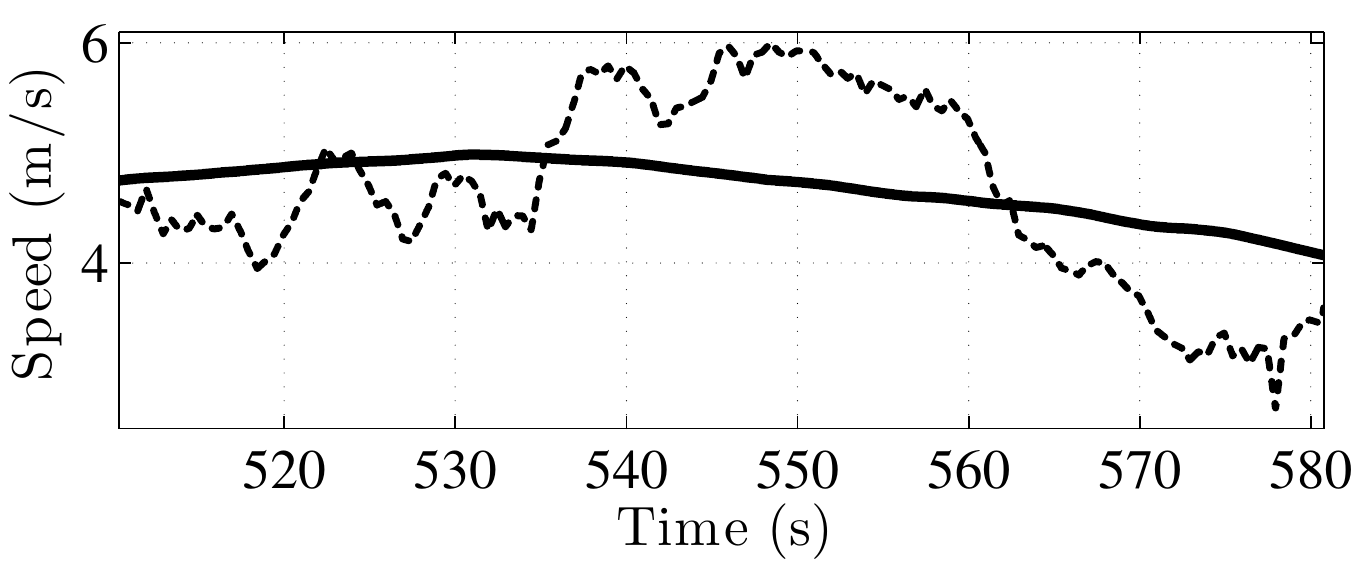}
   \caption[]{Time courses of the wind speed (dotted) and the \minu{1} average wind speed (solid).}
   \label{fig:Exp_Wspd_PwrCycle_GAreg}
  \end{center}
 \end{subfigure}
 \\
 \vspace{0.5cm}
 \begin{subfigure}[b]{0.49\columnwidth}
  \begin{center}
   \includegraphics[trim= 0cm 0cm 0cm 0cm,width=\columnwidth]{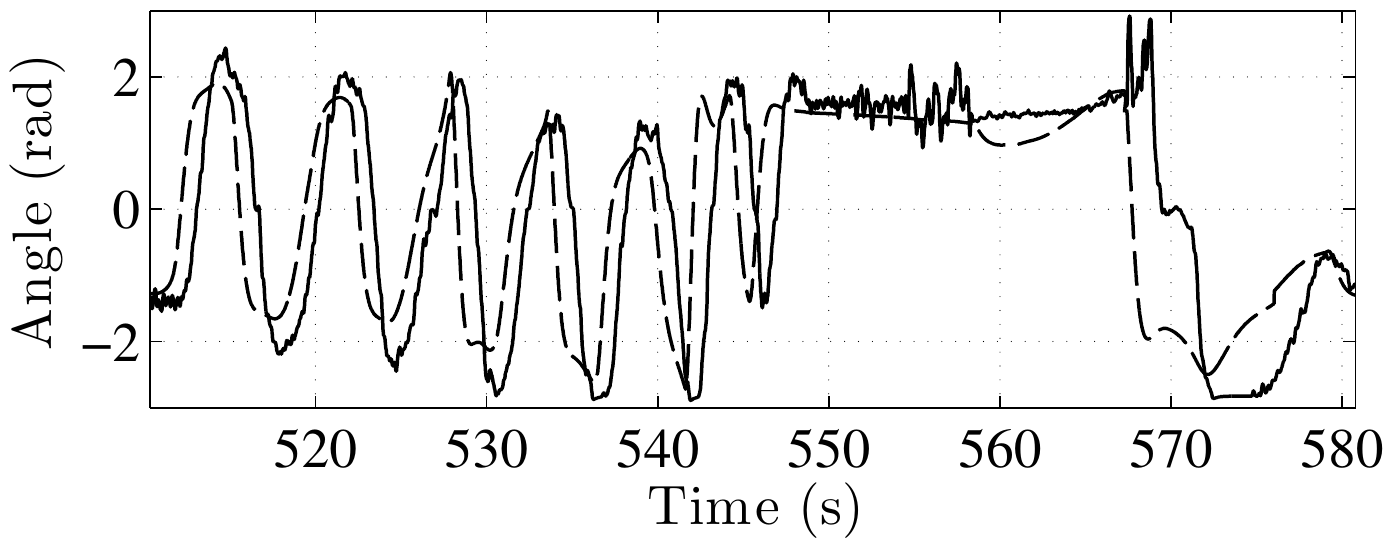}
   \caption[]{Time courses of $\gamma$ (solid) and $\gamma_{\text{ref}}$ (dashed). At roughly $t\in[\seconds{558},\seconds{567}]$, the regularized version of $\gamma$ \eqref{eqn:Gamma_reg} is used for feedback control.}
   \label{fig:Exp_Gamma_PwrCycle_GAreg}
  \end{center}
 \end{subfigure}
 \hfill
 \begin{subfigure}[b]{0.49\columnwidth}
  \begin{center}
   \includegraphics[trim= 0cm 0cm 0cm 0cm,width=\columnwidth]{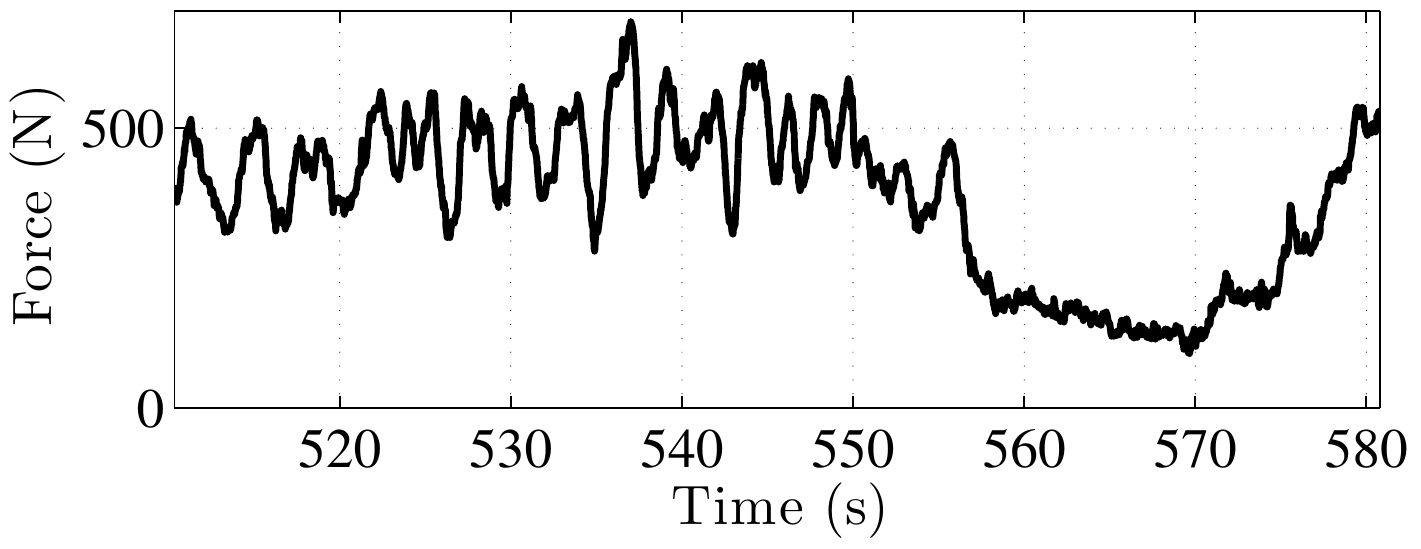}
   \caption[]{Time course of the traction force on the main line (solid).}
   \label{fig:Exp_Fmain_PwrCycle_GAreg}
  \end{center}
 \end{subfigure}
 \\
 \vspace{0.5cm}
 \begin{subfigure}[b]{0.49\columnwidth}
  \begin{center}
   \includegraphics[trim= 0cm 0cm 0cm 0cm,width=\columnwidth]{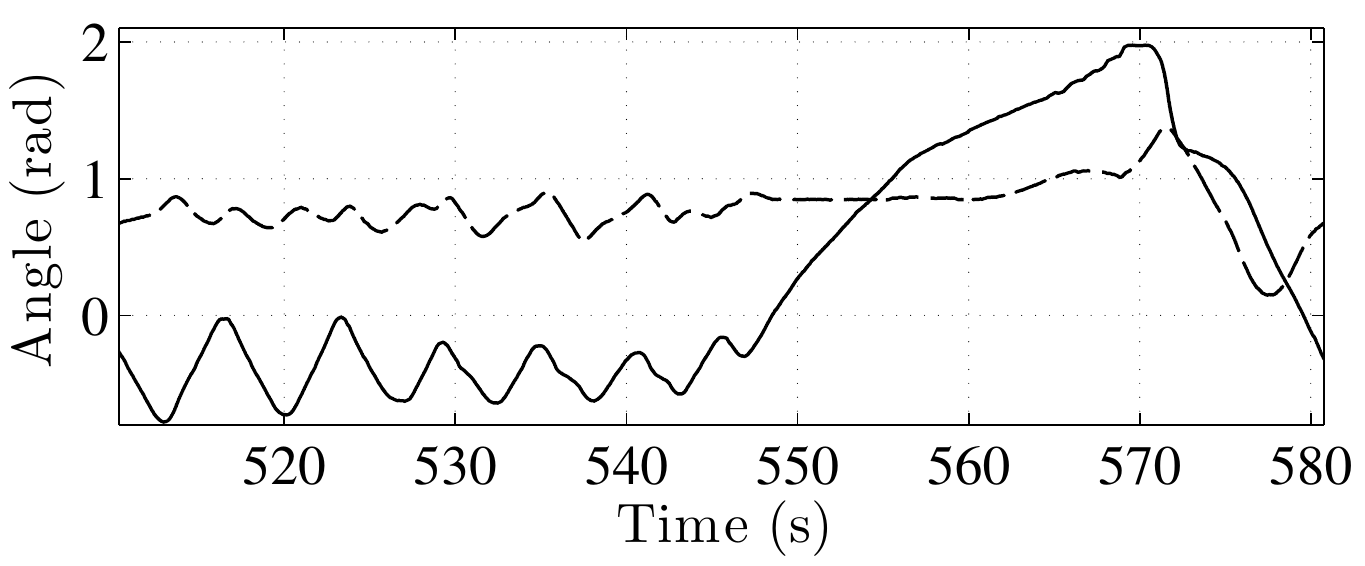}
   \caption[]{Time courses of $\phi$ (solid) and $\theta$ (dashed).}
   \label{fig:Exp_PHTH_PwrCycle_GAreg}
  \end{center}
 \end{subfigure}
 \hfill
 \begin{subfigure}[b]{0.49\columnwidth}
  \begin{center}
   \includegraphics[trim= 0cm 0cm 0cm 0cm,width=\columnwidth]{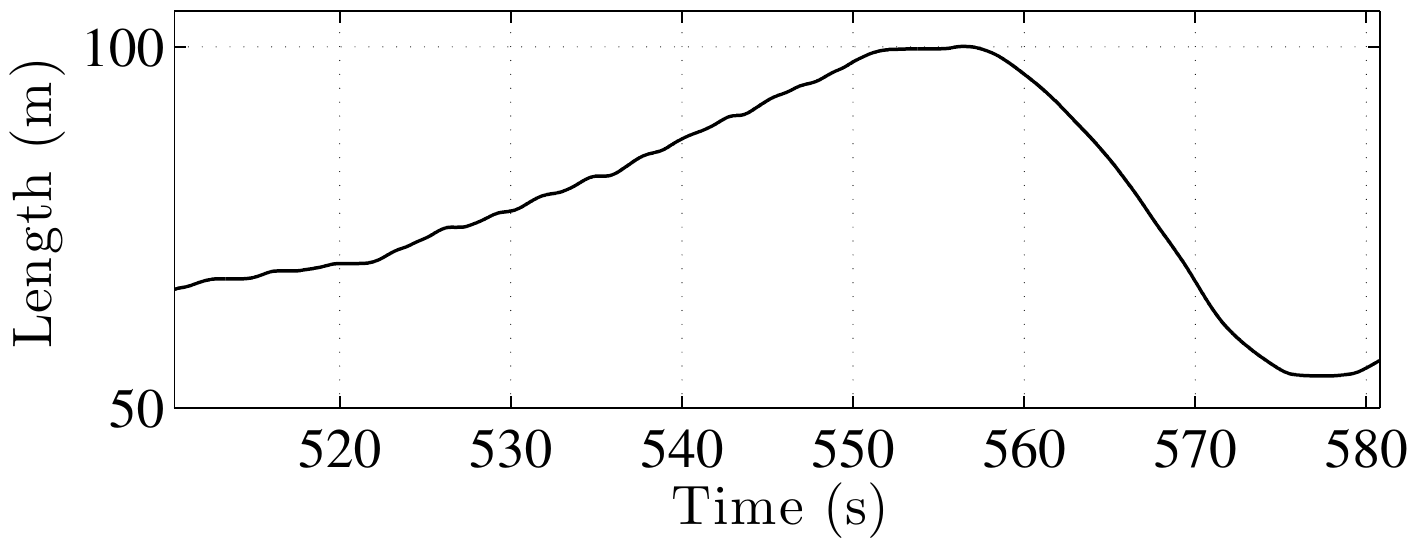}
   \caption[]{Time course of tether length $r$ (solid).}
   \label{fig:Exp_R_PwrCycle_GAreg}
  \end{center}
 \end{subfigure}
 \caption{Experimental results using the velocity angle based retraction with an \airush \sqm{9} kite during one power cycle.}
 \label{fig:Exp_PwrCycle_GAreg}
\end{figure}

\begin{figure}[tbh]
 \begin{center}
  \includegraphics[trim= 0cm 0cm 0cm 0cm,width=0.5\columnwidth]{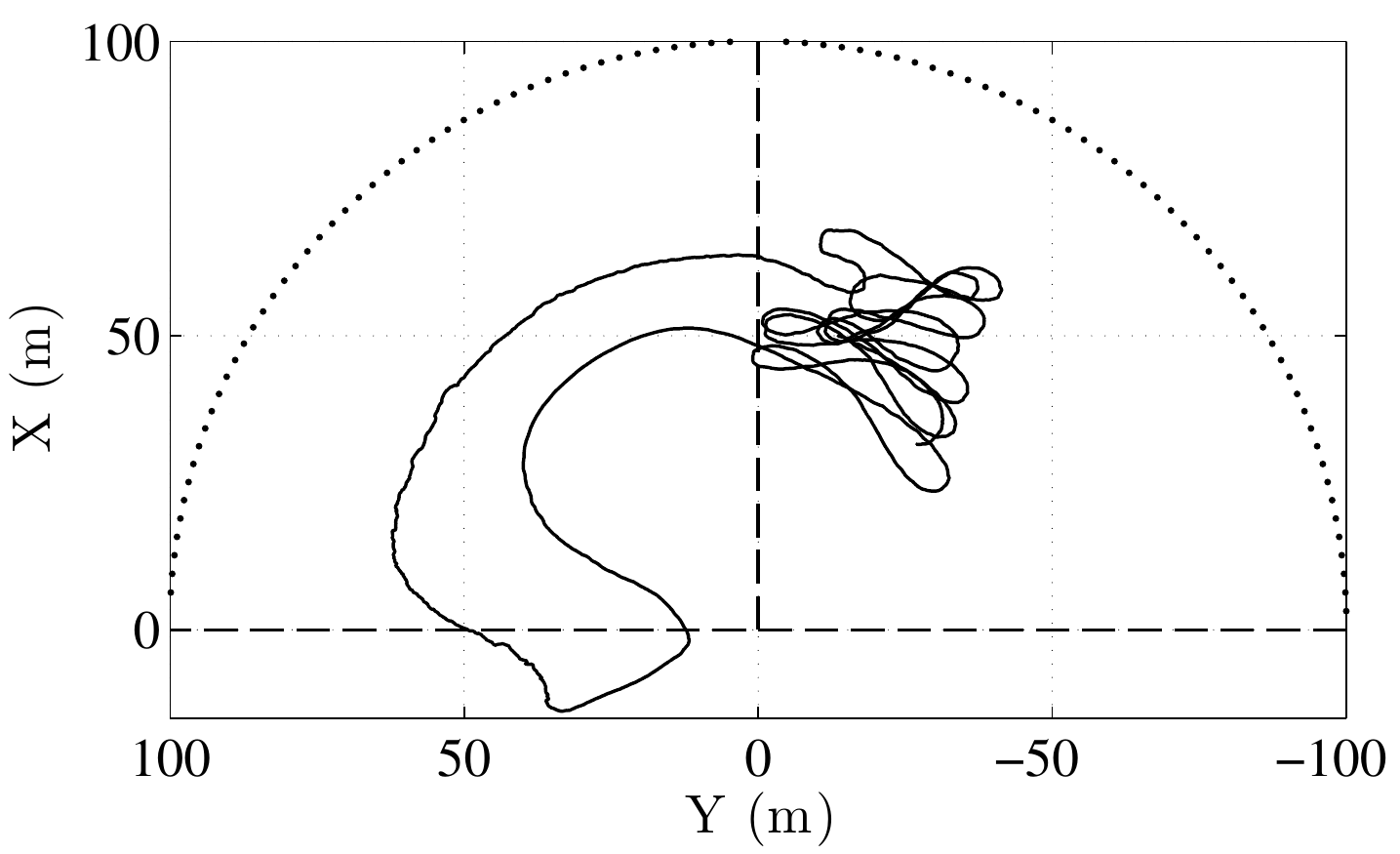}
  \caption[]{Experimental results. Wing path projected on to the ground, corresponding to \figref{fig:Exp_PHTHTraj_PwrCycle_GAreg}.}
  \label{fig:Exp_3DTrajProjXY_PwrCycle_GAreg}
 \end{center}
\end{figure}

Figure~\ref{fig:Exp_PwrCycle_ElCtrl} shows the results of experimental test flights, where the approach based on the elevation dynamics has been used, with the same \airush $\sqm{9}$ kite.

In \figref{fig:Exp_PHTHTraj_PwrCycle_ElCtrl}, the wing path during a power cycle in the $(\phi,\theta)$-plane is shown. 
Again, the retraction controller stabilizes the wing at the border of the wind window during the reel-in of the tether.
Once at the minimum tether length, the wing turns back to fly towards a downwind position using the traction controller \cite{FaZg13}.
In \figref{fig:Exp_PHTH_PwrCycle_ElCtrl} the corresponding time courses of $\phi$ and $\theta$ during the power cycle are shown. One can see that the elevation-based retraction controller corrects the low $\theta$ position of the wing (starting roughly at \seconds{510}) towards $\theta_{\text{ref}}=\radi{1}$.
The wind speed was approximately \mps{5}, see \figref{fig:Exp_Wspd_PwrCycle_ElCtrl}. The corresponding traction force on the main line is visible in \figref{fig:Exp_Fmain_PwrCycle_ElCtrl} and it can be seen that during the   retraction phase the force drops by a factor of two to three.
The resulting tether length during the power cycle is shown in \figref{fig:Exp_R_PwrCycle_ElCtrl}.

There are a few notable differences between \figref{fig:Exp_PHTHTraj_PwrCycle_ElCtrl} and \figref{fig:Exp_PHTHTraj_PwrCycle_GAreg}. To decrease the traction force on the lines and the $\phi$ position overshoot behind the GU against the wind, the pitch angle of the wing was slightly increased in the experiment shown in \figref{fig:Exp_PHTHTraj_PwrCycle_ElCtrl} compared to \figref{fig:Exp_PHTHTraj_PwrCycle_GAreg}, resulting in a lower efficiency of the wing.  Additionally, the gain $K_c$ was kept at a higher value once the new traction phase starts until the wing is in a downwind position (see \tabref{tab:control_params_Ctrl}). This compensates the decrease of the wing's steering gain (see \eqref{eqn:GammaDotLawCoeffs}) and prevents the wing from flying to a low elevation once the traction phase starts, compare \figref{fig:Exp_PHTHTraj_PwrCycle_GAreg} and \figref{fig:Exp_PHTHTraj_PwrCycle_ElCtrl}. Also, as we expected, the wing elevation based retraction controller shows a better tracking performance for $\theta$. This can clearly be seen in \figref{fig:Exp_PHTHTraj_PwrCycle_ElCtrl} and \figref{fig:Exp_PHTH_PwrCycle_ElCtrl} (around 510-\seconds{520}).

\begin{figure}[t]
 \centering
 \begin{subfigure}[b]{0.49\columnwidth}
  \begin{center}
   \includegraphics[trim= 0cm 0cm 0cm 0cm,width=\columnwidth]{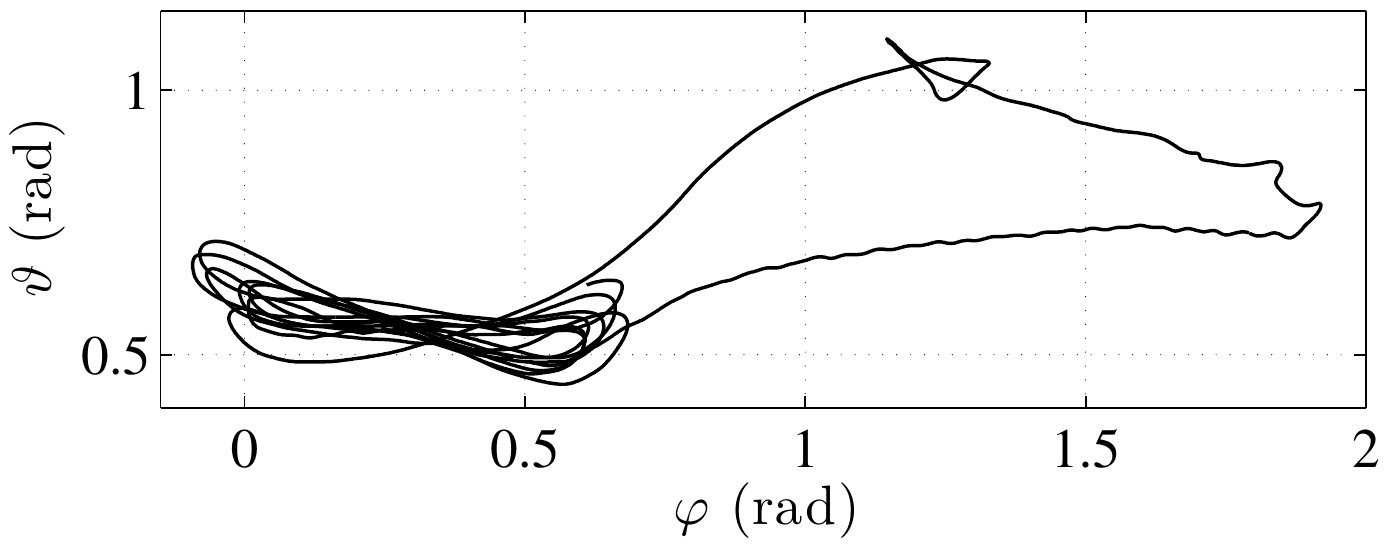}
   \caption[]{The wing trajectory of one power cycle is shown in $\phi$ and $\theta$. The wind direction was roughly $\phi_W\approx\radi{+0.5}$.}
   \label{fig:Exp_PHTHTraj_PwrCycle_ElCtrl}
  \end{center}
 \end{subfigure}
 \hfill
 \begin{subfigure}[b]{0.49\columnwidth}
  \begin{center}
   \includegraphics[trim= 0cm 0cm 0cm 0cm,width=\columnwidth]{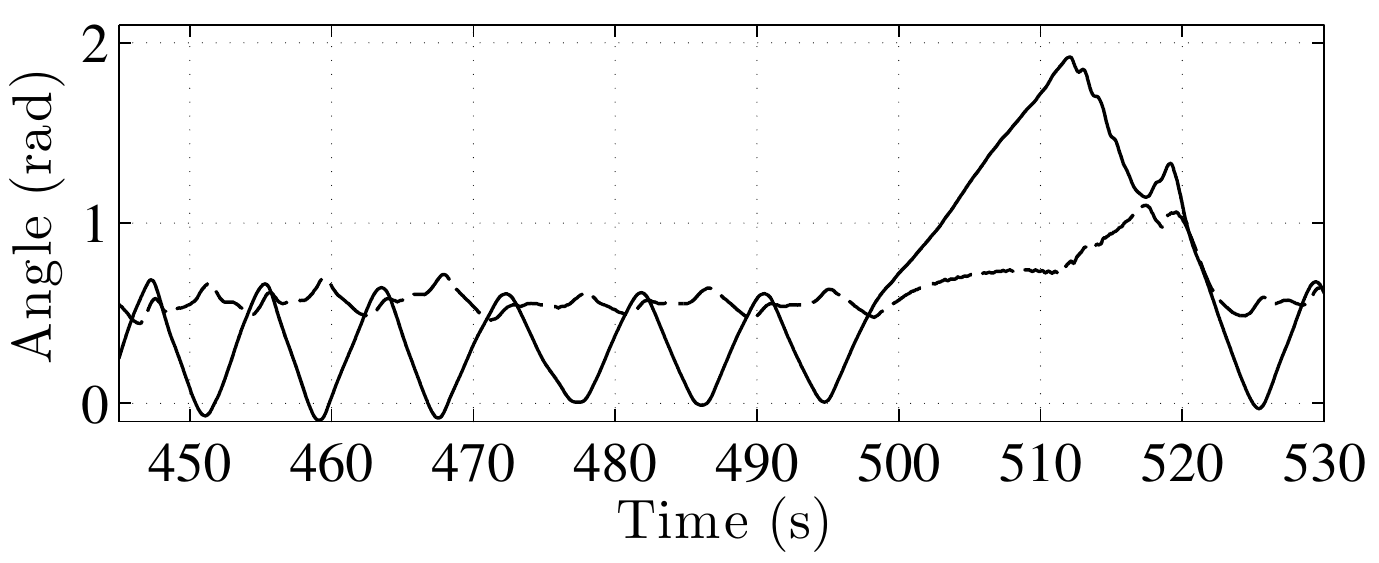}
   \caption[]{Time courses of $\phi$ (solid) and $\theta$ (dashed).}
   \label{fig:Exp_PHTH_PwrCycle_ElCtrl}
  \end{center}
 \end{subfigure}
 \\
 \vspace{0.5cm}
 \begin{subfigure}[b]{0.49\columnwidth}
  \begin{center}
   \includegraphics[trim= 0cm 0cm 0cm 0cm,width=\columnwidth]{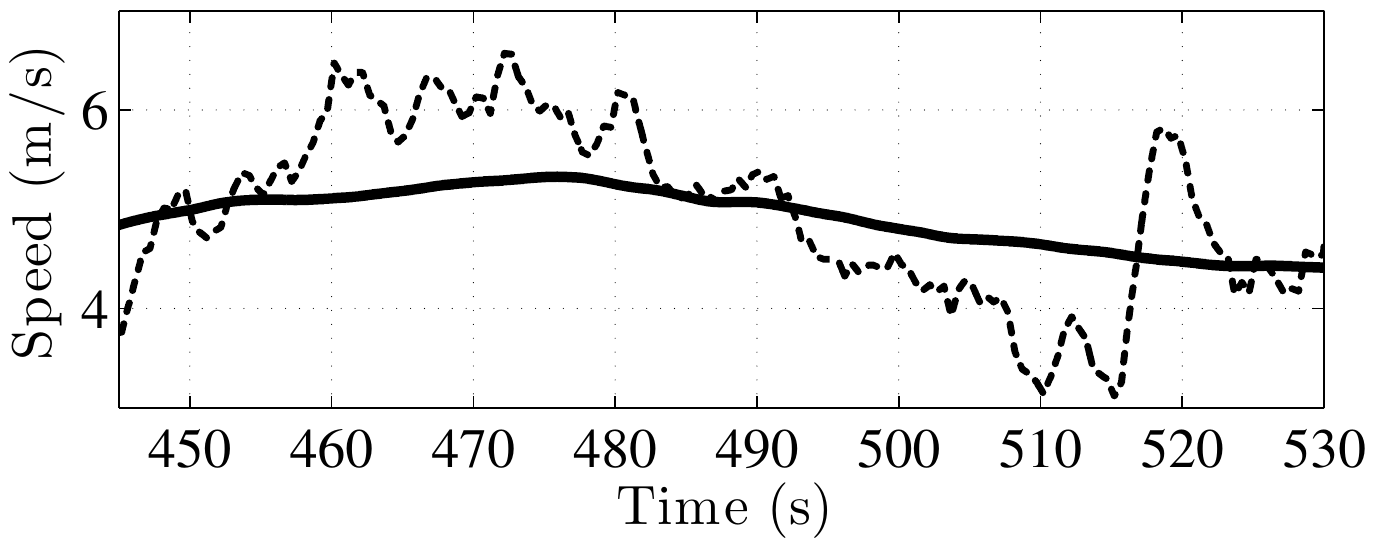}
   \caption[]{Time courses of the wind speed (dotted) and the \minu{1} average wind speed (solid).}
   \label{fig:Exp_Wspd_PwrCycle_ElCtrl}
  \end{center}
 \end{subfigure}
 \hfill
 \begin{subfigure}[b]{0.49\columnwidth}
  \begin{center}
   \includegraphics[trim= 0cm 0cm 0cm 0cm,width=\columnwidth]{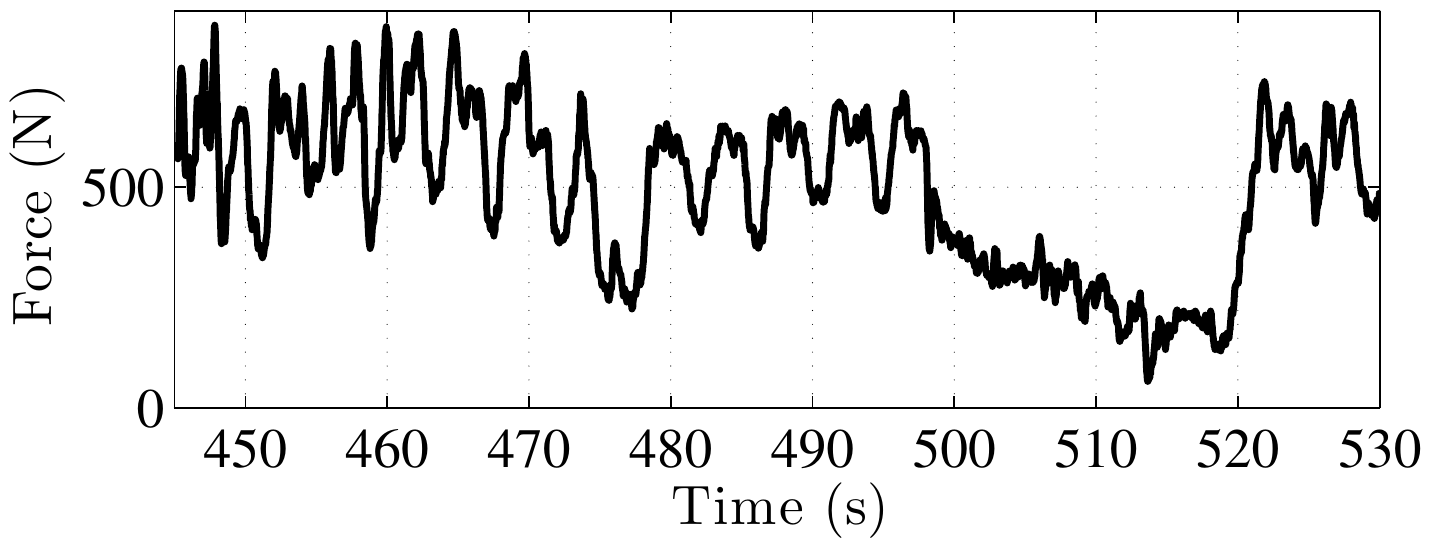}
   \caption[]{Time course of the traction force on the main line (solid).}
   \label{fig:Exp_Fmain_PwrCycle_ElCtrl}
  \end{center}
 \end{subfigure}
 \\
 \vspace{0.5cm}
 \begin{subfigure}[b]{0.49\columnwidth}
  \begin{center}
   \includegraphics[trim= 0cm 0cm 0cm 0cm,width=\columnwidth]{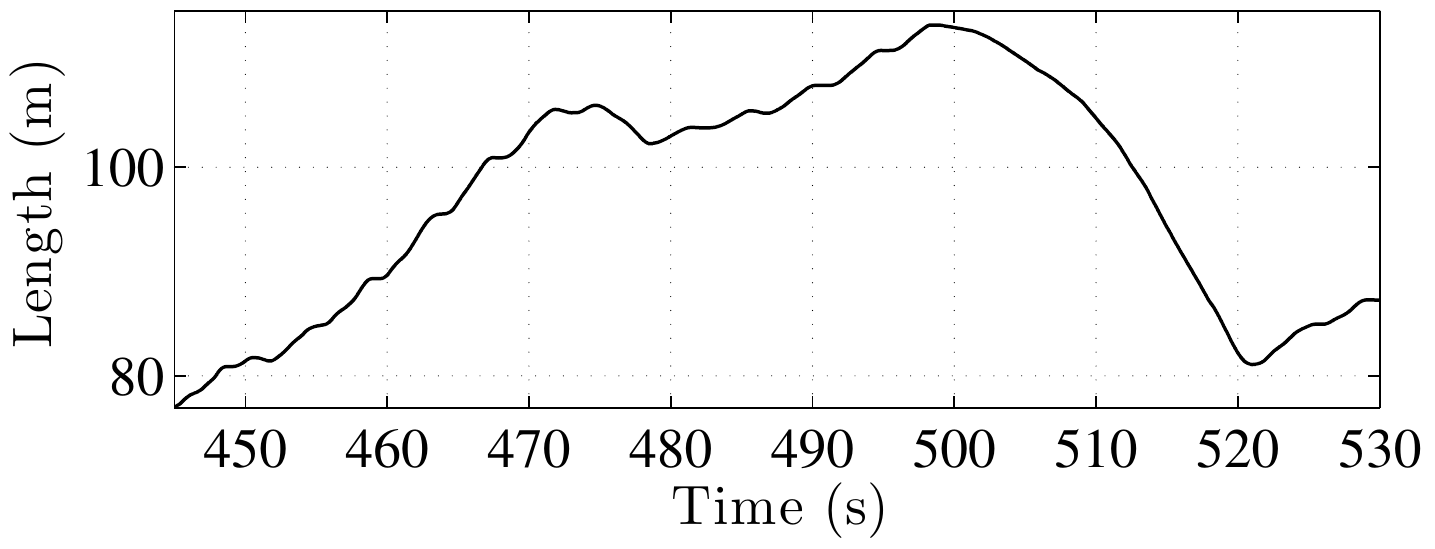}
   \caption[]{Time course of tether length $r$ (solid).}
   \label{fig:Exp_R_PwrCycle_ElCtrl}
  \end{center}
 \end{subfigure}
 \caption{Experimental results using the elevation dynamics based retraction with an \airush \sqm{9} kite during one power cycle.}
 \label{fig:Exp_PwrCycle_ElCtrl}
\end{figure}

In \figsref{fig:Exp_Fmain_PwrCycle_GAreg_wCWTF} and \ref{fig:Exp_Fmain_PwrCycle_ElCtrl_wCWTF} a comparison of the traction force between the actual measurements during one power cycle and the simplified traction force model \eqref{eqn:CWTFmodel} are shown. 
In \figref{fig:Exp_Pmain_PwrCycle_GAreg_wCWTF} and \figref{fig:Exp_Pmain_PwrCycle_ElCtrl_wCWTF} a comparion of the mechanical power on the main line is shown using the same model.
This model has been widely used to estimate and optimize the power output of an AWE system, as well as to carry out economical considerations for AWE generators.
To carry out such a comparison, the lift coefficient and equivalent efficiency where estimated using a fraction of the data set. These values can change even for the same wing if different bridling setups are used. In \figref{fig:Exp_Fmain_PwrCycle_GAreg_wCWTF} the values are $C_L=0.8$ and $E_{eq}=3.7$ whereas in \figref{fig:Exp_Fmain_PwrCycle_ElCtrl_wCWTF} they are $C_L=0.8$ and $E_{eq}=3.2$.
Additionally, we do not consider a wind shear effect since it is difficult to estimate the wind shear with only a ground based anemometer available on our prototype. Therefore we assume that the wind speed measured roughly \Meter{5} above the ground corresponds to the wind speed at the wing's location. This generally leads to an underestimate of the traction force.
The two plots in \figsref{fig:Exp_Fmain_PwrCycle_GAreg_wCWTF}-\ref{fig:Exp_Fmain_PwrCycle_ElCtrl_wCWTF} show a good correspondence during the traction phase with the tendency of slightly underestimating the traction force. During the retraction phase the assumptions made in \cite{FaMP11} do not hold anymore and the model tends to a larger deviation, compare \figref{fig:Exp_Fmain_PwrCycle_GAreg_wCWTF}. In \figref{fig:Exp_Fmain_PwrCycle_ElCtrl_wCWTF}, a drop in wind speed and a lower reel-in speed led to a good matching of the traction force during the retraction phase.
The spike of the force in \figref{fig:Exp_Fmain_PwrCycle_ElCtrl_wCWTF} given by the model at roughly \seconds{475} can not exactly be explained by the data but can be caused by a wind gust at the wing's location, which is not seen by the ground based anemometer, leading to a reel-in during the traction phase to keep a minimum tension on the lines, compare \figsref{fig:Exp_R_PwrCycle_ElCtrl} and \ref{fig:Exp_Wspd_PwrCycle_ElCtrl}. This reel-in speed increases $\vec{W}_a^r$ in the model while still assuming the same wind speed and thus leading to a too high force estimate.
The time course of the mechanical power on the main line is compared to the simplified model in \figref{fig:Exp_Pmain_PwrCycle_GAreg_wCWTF} and \figref{fig:Exp_Pmain_PwrCycle_ElCtrl_wCWTF}. It can be noted that the average power values are quite consistent during the traction phase, and that the simplified model is subject to lower variability during the traction phase since it does not consider the changing wing speeds during the figure eight pattern.

\begin{figure}[t]
 \begin{subfigure}[b]{0.49\columnwidth}
  \begin{center}
   \includegraphics[trim= 0cm 0cm 0cm 0cm,width=\columnwidth]{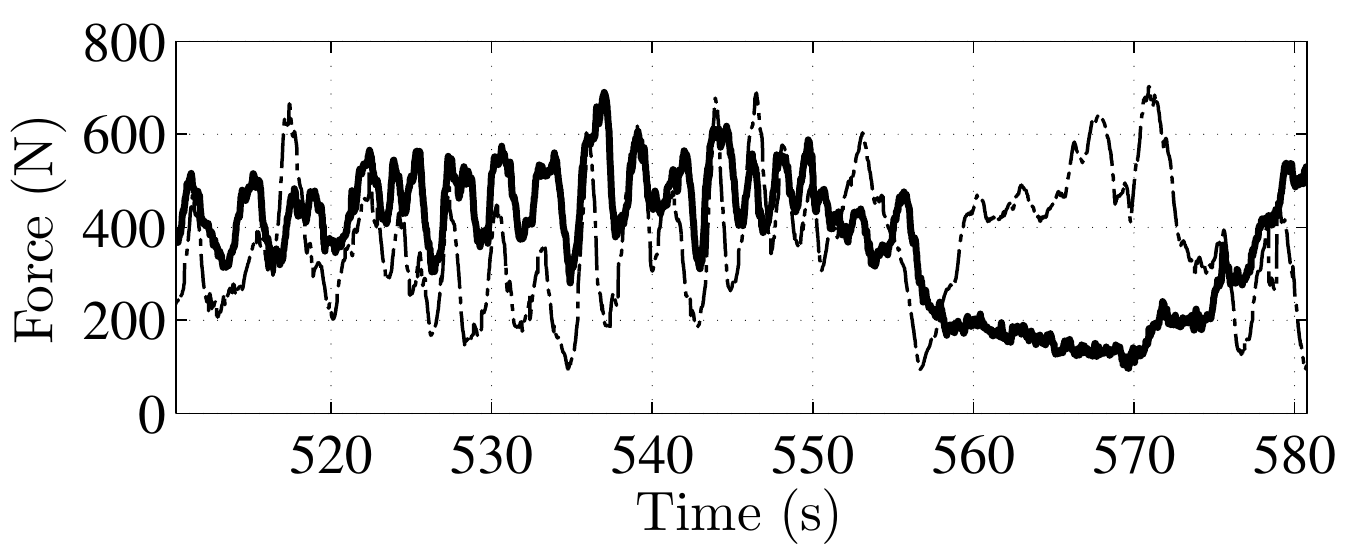}
   \caption[]{Time courses of the traction force on the main line (solid), corresponding to \figref{fig:Exp_PwrCycle_GAreg}, and the traction force model \eqref{eqn:CWTFmodel} (dot-dashed) .}
   \label{fig:Exp_Fmain_PwrCycle_GAreg_wCWTF}
  \end{center}
 \end{subfigure}
 \hfill
 \begin{subfigure}[b]{0.49\columnwidth}
  \begin{center}
   \includegraphics[trim= 0cm 0cm 0cm 0cm,width=\columnwidth]{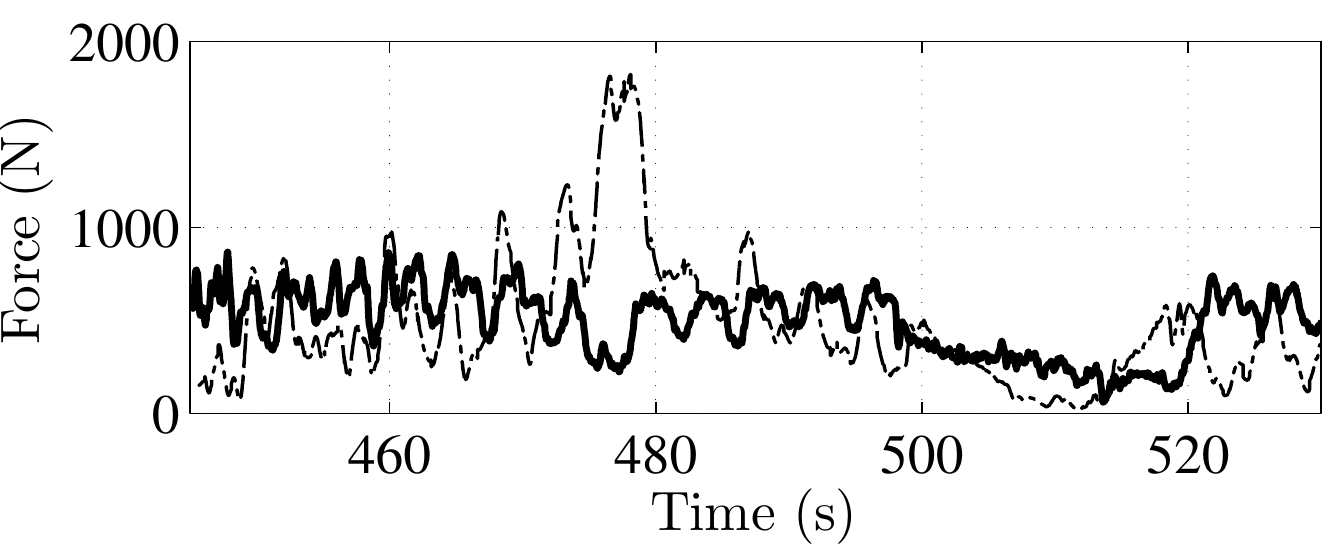}
   \caption[]{Time courses of the traction force on the main line (solid), corresponding to \figref{fig:Exp_PwrCycle_ElCtrl}, and the traction force model \eqref{eqn:CWTFmodel} (dot-dashed).}
   \label{fig:Exp_Fmain_PwrCycle_ElCtrl_wCWTF}
  \end{center}
 \end{subfigure}
 \\
 \vspace{0.5cm}
 \begin{subfigure}[b]{0.49\columnwidth}
  \begin{center}
   \includegraphics[trim= 0cm 0cm 0cm 0cm,width=\columnwidth]{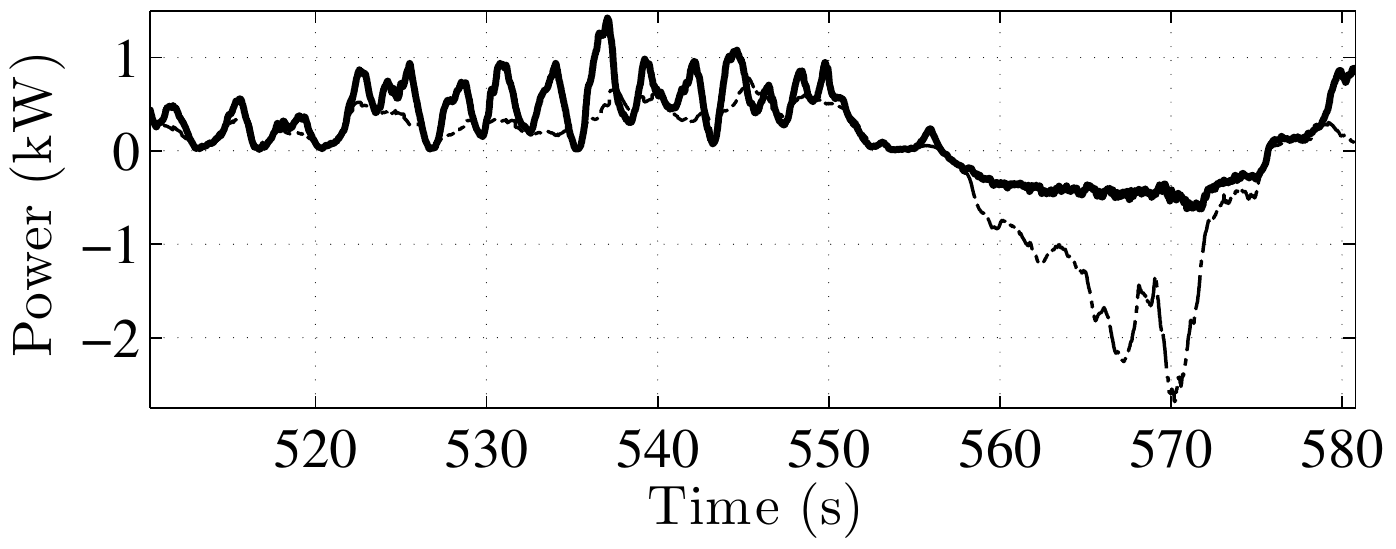}
   \caption[]{Time courses of the mechanical power on the main line (solid), corresponding to \figref{fig:Exp_PwrCycle_GAreg}, and the traction force model \eqref{eqn:CWTFmodel} (dot-dashed) .}
   \label{fig:Exp_Pmain_PwrCycle_GAreg_wCWTF}
  \end{center}
 \end{subfigure}
 \hfill
 \begin{subfigure}[b]{0.49\columnwidth}
  \begin{center}
   \includegraphics[trim= 0cm 0cm 0cm 0cm,width=\columnwidth]{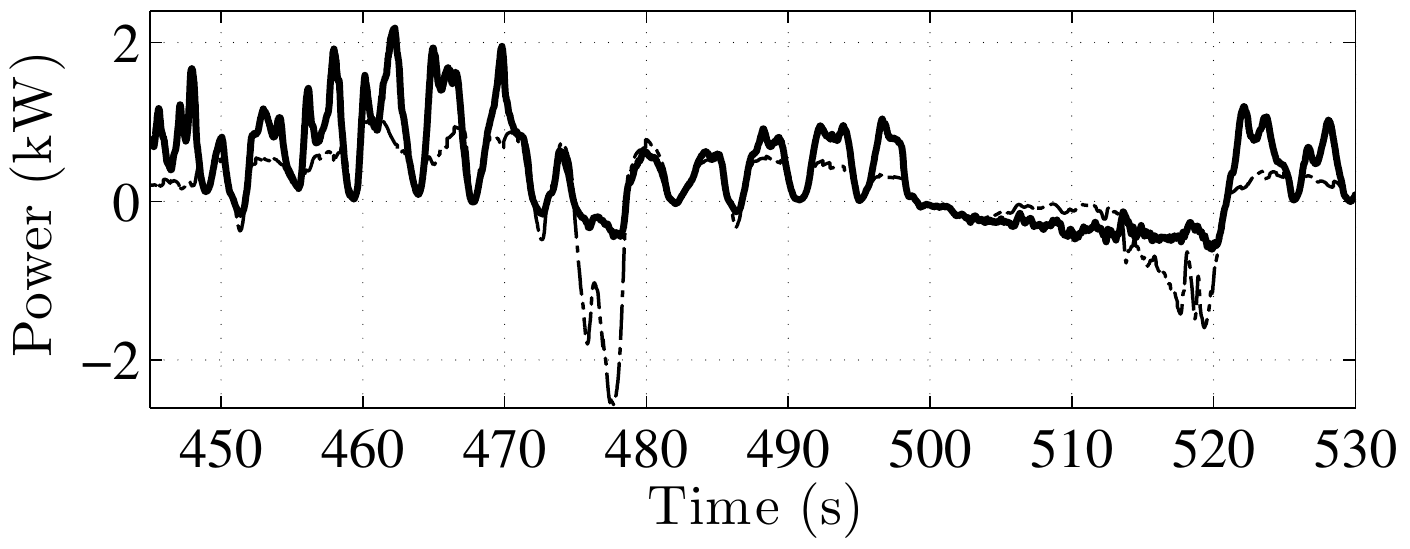}
   \caption[]{Time courses of the mechanical power on the main line (solid), corresponding to \figref{fig:Exp_PwrCycle_ElCtrl}, and the traction force model \eqref{eqn:CWTFmodel} (dot-dashed).}
   \label{fig:Exp_Pmain_PwrCycle_ElCtrl_wCWTF}
  \end{center}
 \end{subfigure}
 \caption{Comparison of experimental data, corresponding to \figsref{fig:Exp_PwrCycle_GAreg} and \ref{fig:Exp_PwrCycle_ElCtrl}, with the simplified traction force model given in \eqref{eqn:CWTFmodel}}
 \label{fig:Exp_PwrCycle_wCWTF}
\end{figure}


%% file: conclusion.tex
\section{Conclusion}
\label{sec:conclusion}
We proposed two different approaches to design a feedback controller for the retraction phase of an AWE system with ground-based generation, where the tether is recoiled onto the drums. Together with a previously proposed traction controller, and with a torque-based reeling control strategy, the approaches presented here have been used to achieve fully autonomous power cycles.

The two approaches were compared in simulation employing a nonlinear point-mass model and in real experiments using the Swiss Kite Power prototype. Both approaches were able to stabilize the wing in a position at the border of the wind window and can be used to fly complete power cycles with a tethered wing. The approach based on the elevation dynamics is more promising since it relies only on directly measured variables.

For both approaches, only few parameters, that can be intuitively tuned, are involved in the design. The approaches employ the steering deviation as control input and can stabilize the wing's elevation robustly against different tether lengths and reeling speeds. Hence, the latter can still be optimized to maximize the energy output of the system.

The presented automatic controllers for the retraction phase are two among the few which have so far been proven to work on real prototypes. 
Future research on this topic can be devoted to the power cycle optimization by using multivariable approaches and to the inclusion of active pitch strategies.

%% file: appendix.tex

\section*{Proof of Proposition \ref{prop:Elevation_Dynamics}}
\label{App:ProofOfProp}
  \label{proof:TH_dyn_approx}
  For the sake of simplicity of notation we drop the dependence of time-varying variables on $t$.
 
  The components of the force $\vec{F}$ in \eqref{eqn:EqnMot} in $\theta$ direction are given by the gravitational force $\vec{F}_g$ and the aerodynamic force $\vec{F}_a$.
  The gravitational force can be expressed in the local frame L as:
  \begin{IEEEeqnarray}{rCl}\label{eqn:APP_L_Fg}
   {}_L\vec{F}_g &=& \begin{pmatrix}
                      -mg\cos{(\theta)}\\
                       0\\
                       mg\sin{(\theta)}
                     \end{pmatrix}
  \end{IEEEeqnarray}
  with $m$ being the mass of the wing plus the added mass of the tether and $g$ is the gravitational acceleration.
  The aerodynamic force is given as:
  \begin{IEEEeqnarray}{rCl}
   \vec{F}_a &=& F_L\,\vec{e}_L+F_{D,eq}\,\vec{e}_W\, ,\label{eqn:APP_F_aero_general}
  \end{IEEEeqnarray}
  where $F_L$ is the lift force and $F_{D,eq}$ the equivalent drag force including also the tether drag:
  \begin{IEEEeqnarray}{rCl}
  F_L      &=& \frac{1}{2}\rho A C_L |\vec{W}_a|^2\label{eqn:F_L_areo}\\
  F_{D,eq} &=& \frac{1}{2}\rho A C_{D,eq} |\vec{W}_a|^2\, .\label{eqn:F_D_areo}
  \end{IEEEeqnarray}
  In \eqref{eqn:F_L_areo} and \eqref{eqn:F_D_areo}, $\rho$ is the air density, $A$ is the effective area of the wing, $C_L$ and $C_{D,eq}$ are the lift coefficient and equivalent drag coefficient, and $\vec{W}_a$ is the apparent wind velocity.
  The vectors $\vec{e}_L$ and $\vec{e}_W$ in \eqref{eqn:APP_F_aero_general} can be expressed in the $L$ frame as:
   \begin{IEEEeqnarray}{rCl}
    {}_L\vec{e}_L(t) &=& \begin{pmatrix}
                          \cos{\xi} & -\sin{\xi} & 0\\
                          \sin{\xi} &  \cos{\xi} & 0\\
                          0         &  0         & 1
                         \end{pmatrix}\cdot\nonumber\\
                     & & \begin{pmatrix}
                           \cos{\psi}\cos{\eta}\sin{\Delta\alpha}\\
                           \cos{\psi}\sin{\eta}\sin{\Delta\alpha}+\sin{\psi}\cos{\Delta\alpha}\\
                          -\cos{\psi}\cos{\eta}\cos{\Delta\alpha}
                         \end{pmatrix}\label{eqn:APP_L_e_L}\\
    {}_L\vec{e}_W(t) &=& \begin{pmatrix}
                          \cos{\xi} & -\sin{\xi} & 0\\
                          \sin{\xi} &  \cos{\xi} & 0\\
                          0         &  0         & 1
                         \end{pmatrix}
                         \begin{pmatrix}
                          -\cos{\Delta\alpha}\\
                           0\\
                          -\sin{\Delta\alpha}
                         \end{pmatrix}\,,\label{eqn:APP_L_e_W}
   \end{IEEEeqnarray}
  where $\Delta\alpha$ is the angle between the apparent wind and the tangent plane $(\vec{e}_N,\vec{e}_E)$, $\psi$ the roll angle of the wing which is a function of the steering input $\delta$:
  \begin{IEEEeqnarray}{rCl}\label{eqn:APP_psi_definition}
   \psi &=& \arcsin{\left(\frac{\delta}{d_s}\right)}\, ,
  \end{IEEEeqnarray}
  $\eta$ is given by (see e.g. \cite{Argatov2009}):
  \begin{IEEEeqnarray}{rCl}\label{eqn:APP_eta_eq}
   \eta &=& \arcsin{(\tan{(\Delta\alpha)}\tan{(\psi)})}\, ,
  \end{IEEEeqnarray}
  and $\xi$ is the heading of the wing which is given by the apparent wind vector $\vec{W}_a$, defined in \eqref{eqn:WapparentDefinition}, and can be written as:
  \begin{IEEEeqnarray}{rCl}\label{eqn:xi_def}
  \xi &=& \arctan{\left(\frac{-\vec{W}_a\cdot\mathbf{e}_{E}}{-\vec{W}_a\cdot\mathbf{e}_{N}}\right)}\label{eqn:APP_Wing_Heading_Def}\\
      &=&  \arctan{\left(\frac{W_0\sin{(\phi)}+r\cos{(\theta)}\dot{\phi}}{W_0\sin{(\theta)}\cos{(\phi)}+r\dot{\theta}}\right)}\, .
  \end{IEEEeqnarray}
  The assumption underlying equation \eqref{eqn:APP_Wing_Heading_Def} is that the wing's longitudinal symmetry axis is always contained in the plane spanned by the vectors $\vec{W}_a$ and $\vec{p}$ and is common in the field of AWE \cite{HouDi07,IlHD07,WillLO08}.
  
  Thus the force $\mathbf{F}$ in $\mathbf{e}_N$ direction can be computed as:
    \begin{IEEEeqnarray}{rCl}\label{eqn:F_e_N}
  \mathbf{F}\cdot\mathbf{e}_N &=& F_L\left(\cos{(\eta)}\sin{(\Delta\alpha)}\cos{(\xi)}-\right.\nonumber\\
			      & & \phantom{F_L}\left.\left(\sin{(\eta)}\sin{(\Delta\alpha)}+\cos{(\Delta\alpha)}\psi\right)\sin{(\xi)}\right)-\\
			      &_& F_{D,eq}\cos{(\Delta\alpha)}\cos{(\xi)} - mg\cos{(\theta)}\, .\nonumber
  \end{IEEEeqnarray}
  For more details and a formal definition of the components of $\vec{F}$ see e.g. \cite{FaZg13}.

  By \asref{as:SteadyState} and considering the equilibrium of the lift and drag force in the direction of the wing's heading $\xi$, projected on the tangent plane to the wind window at the wing's location, we have (see \cite{FaMP11}):
  \begin{IEEEeqnarray}{rCl}
  \label{eqn:APP_LoverDEq}
  \frac{\sin{(\Delta\alpha)}}{\cos{(\Delta\alpha)}} &=& \frac{C_{D,eq}}{C_L} \doteq \frac{1}{E_{eq}}\, ,
  \end{IEEEeqnarray}
  where $E_{eq}$ is the equivalent efficiency of the wing. By \eqref{eqn:APP_LoverDEq} we can see that $\Delta\alpha$ is small for a reasonable wing efficiency of $4-6$.

  By \asref{as:SmallInput}, \eqref{eqn:APP_psi_definition}, and \eqref{eqn:APP_LoverDEq} we see that \eqref{eqn:APP_eta_eq} simplifies to
  \begin{IEEEeqnarray}{rCl}\label{eqn:eta_linear}
  \eta = \frac{1}{E_{eq}}\psi =\frac{1}{E_{eq}d_s}\delta \, ,
  \end{IEEEeqnarray}
  where $d_s$ is the span of the wing.
  
  By using \eqref{eqn:F_e_N} together with \eqref{eqn:F_L_areo}-\eqref{eqn:xi_def} we obtain:
  \begin{IEEEeqnarray}{rCl}
  \mathbf{F}\cdot\mathbf{e}_N &=& \frac{\rho AC_L}{2d_s}\left(1+\frac{1}{E_{eq}^2}\right)(W_0\sin{(\phi)}+r\cos{(\theta)}\dot{\phi})|\vec{W}_a|\delta-\nonumber\\
			      & & mg\cos{(\theta)}\, .
  \end{IEEEeqnarray}
  Equation \eqref{eqn:THdd_dyn}, by \asref{as:SteadyState}, can now be rewritten as
  
  \begin{IEEEeqnarray}{rCl}\label{eqn:TH_dyn_approx_appendix}
  \ddot{\theta}               &=& -\mathcal{C}\delta-\frac{g\cos{(\theta)}+2\dot{\theta}\dot{r}}{r}\, ,
  \end{IEEEeqnarray}
  
  where
  \begin{IEEEeqnarray}{rCl}
  \mathcal{C}                 &=& \frac{\rho AC_L}{2rmd_s}\left(1+\frac{1}{E_{eq}^2}\right)W_0\sin{(\phi)}|\vec{W}_a|\, .
  \end{IEEEeqnarray}
  \hfill$\blacksquare$